\documentclass[aip,jmp,showpacs,superscriptaddress]{revtex4-1}
\usepackage{graphicx}
\usepackage{caption}
\usepackage{subcaption}
\usepackage{amsfonts}
\usepackage{dsfont}
\usepackage{amsmath,amsthm}
\usepackage{amssymb}
\usepackage{mathrsfs}
\usepackage{amsmath, amsthm, amssymb,amsfonts,mathbbol,amstext}
\usepackage{mathtools}
\usepackage[colorlinks=true,citecolor=blue,urlcolor=blue]{hyperref}
\usepackage[normalem]{ulem}
\usepackage{pgf,tikz}\usepackage{mathrsfs}\usetikzlibrary{arrows}

\newtheorem{dfn}{Definition}

\newtheorem{thm}{Theorem}
\newtheorem{cor}[thm]{Corollary}

\newtheorem{prop}[thm]{Proposition}

\def\be{\begin{equation}}
\def\ee{\end{equation}}
\newcommand{\mc}[1]{\mathcal{#1}}
\usepackage{color}
\definecolor{violeta}{cmyk}{0.07,0.90,0,0.34}

\usepackage{cancel}
\usepackage{soul}
\definecolor{cgreen}{RGB}{26, 199, 76}

\begin{document}

%%%%%%%%%%%%%%%%%%%%%%%%%%%%%%%%%%%%%%%%%%%%%%%%%%%%%%%%%%%%%%%%%%%

\title{Necessary  Conditions for Extended Noncontextuality in General Sets of Random Variables}

\author{Barbara Amaral}
\affiliation{Departamento de F\'isica e Matem\'atica, CAP - Universidade Federal de S\~ao Jo\~ao del-Rei, 36.420-000, Ouro Branco, MG, Brazil} 
\affiliation{International Institute of Physics, Federal University of Rio Grande do Norte, 59078-970, P. O. Box 1613, Natal, Brazil}

\author{Cristhiano Duarte} 
\affiliation{International Institute of Physics, Federal University of Rio 
Grande do Norte, 59078-970, P. O. Box 1613, Natal, Brazil}
\affiliation{Departamento de Matem\'{a}tica, Instituto de Ci\^{e}ncias 
Exatas, Universidade Federal de Minas Gerais, CP 702, CEP 30123-970, Belo 
Horizonte, Minas Gerais, Brazil.}

\author{Roberto I. Oliveira}
\affiliation{Instituto de Matem\'atica Pura e Aplicada (IMPA), Rio de Janeiro, RJ, Brazil.}

\begin{abstract}
We explore  the graph 
approach to contextuality to restate the extended  definition of 
noncontextuality as given by J. Kujala
 et. al. in Ref~\cite{KDL15} in using
graph-theoretical terms.  This 
extended definition avoids the assumption of the  \emph{pre-sheaf} or 
\emph{non-disturbance} condition, which states that if two contexts 
overlap, then the marginal distribution obtained for the intersection must 
be the same, a restriction that will never be perfectly satisfied in real 
experiments. 
 With this we are able to 
 derive necessary conditions for extended noncontextuality for any set of random variables based on the geometrical aspects of the graph approach, which can be tested directly with  experimental data in any 
 contextuality experiment and which reduce to traditional necessary conditions for noncontextuality if the non-disturbance condition is satisfied. 
\end{abstract}

\maketitle

\section{Introduction}

Quantum theory assigns probabilities to subsets of possible measurements of a physical system. The phenomenon of {\em contextuality} states that there may be no 
global probability distribution that is consistent with these subsets, which are also called  \emph{contexts} \cite{Specker60,Bell66,KS67,Fine82,AB11}.

A key consequence of contextuality is that the statistical predictions of quantum theory cannot be obtained from models where the measurement outcomes reveal pre-existent properties that are 
independent on which, or whether, other compatible measurements are 
jointly performed. This fundamental limitation follows from the existence of incompatible  measurements in quantum systems. It thus represents an exotic, intrinsically non-classical phenomenon, that leads to a more fundamental understanding of many aspects of quantum theory \cite{NBAASBC13,Cabello13,Cabello13c,CSW14,ATC14,Amaral14}.
In addition, contextuality has been recognized as  a potential 
\emph{resource} for quantum computing, \cite{Raussendorf13, HWVE14,DGBR14}, random number certification \cite{UZZWYDDK13}, and several other 
information processing tasks in the specific case of space-like separated 
systems \cite{BCPSW13}.

As a consequence, experimental verifications of contextuality have 
received much atten\-tion \cite{HLBBR06, KZGKGCBR09, ARBC09, LLSLRWZ11, BCAFACTP13}. It is thus of utmost
importance to develop a robust theoretical framework for contextuality that 
can be efficiently applied to real experiments. In 
particular, it is important to include the treatment of sets of random 
variables that do not satisfy the assumption of the  so called 
\emph{pre-sheaf}~\cite{AB11,Amaral14} or 
\emph{non-disturbance}~\cite{NBAASBC13} condition. This assumption 
states that if the intersection of two contexts is non-empty, then the 
marginal probability distributions {at  the  intersection  must 
be the same, a restriction that will never be perfectly satisfied in real 
experiments. This problem was considered in Refs. \cite{Larsson02, 
Winter14}, but the methods proposed there  to take into account the 
context-dependent change in a random variable involve quantities that cannot be directly measured.

  In Ref.~\cite{KDL15},the authors propose an alternative definition of 
noncontextuality that can be applied to any set of random 
variables. Such a treatment reduces to the 
traditional definition of noncontextuality if the non-disturbance property 
is satisfied and, in addition, it can be verified directly from  
experimental data. In this alternative definition, a set of random 
variables  is said to be noncontextual (in the extended sense) if there is 
a joint probability distribution which is consistent with the joint 
distribution for each context and maximizes the probability of  two 
realizations of the same set of random variables present in 
different contexts being equal.  Then the authors provide necessary and 
sufficient conditions for contextuality in a broad class of scenarios, 
namely the so called $n$-cycle scenario.

In this contribution, we explore the graph approach to contextuality,
developed in Refs. \cite{CSW10,CSW14, AII06} and further  explored in Refs.
\cite{RDLTC14, AFLS15, AT17}, to rewrite the definition of extended 
noncontextuality in graph theoretical terms.
To this end, from the compatibility graph $\mathrm{G}$ of 
a scenario $\Gamma$, we define another graph $\mathscr{G}$, 
which we call the \emph{extended compatibility graph} of the scenario, and 
show that noncontextuality in the extended sense is equivalent to 
noncontextuality in the traditional sense with respect to the extended 
graph $\mathscr{G}$.

With this graph-theoretical perspective, the problem of characterizing 
extended noncontextuality reduces to  characterizing 
traditional noncontextuality for the scenario defined by $\mathscr{G}$, a 
difficult problem  for general graphs \cite{Pitowsky91, DL97, AII06, 
AII08}. Nevertheless, we can explore the connection between the 
noncontextual set and  the \emph{cut polytope} 
$\mathrm{CUT}\left(\mathrm{G}\right)$ \cite{AII06, AT17} of the 
corresponding compatibility graph $\mathrm{G}$ to derive necessary 
conditions for extended contextuality in any scenario, which can be tested 
directly with  experimental data in any contextuality experiment and  
reduces to traditional necessary conditions for noncontextuality if the 
non-disturbance condition is satisfied.
 
  To derive these conditions, we first prove that $\mathscr{G}$ can be 
obtained from $\mathrm{G}$ combining the  graph operations know as 
\emph{triangular elimination}, \emph{vertex splitting} and \emph{edge 
contraction} \cite{BM86,AII08,Bonato14}. From valid inequalities for 
$\mathrm{CUT}\left(\mathrm{G}\right)$ it is possible to derive valid 
inequalities for any graph obtained from $\mathrm{G}$ using a sequence of 
such operations. In particular, for any valid inequality 
for $\mathrm{CUT}\left(\mathrm{G}\right)$ we can derive valid inequalities 
for $\mathrm{CUT}\left(\mathscr{G}\right)$, among which there is one that 
reduces to the original inequality if the non-disturbance condition is 
satisfied. 

As applications of our framework, we recover the characterization of 
extended noncontextuality for the $n$-cycle scenarios of Ref. \cite{KDL15} 
and provide necessary conditions for noncontextuality exploring the 
$I_{3322}$ \cite{Foissart81,CG04} and Chained inequalities \cite{BC90}. 
Finally, we use the Peres-Mermin square \cite{Peres90,Mermin90} to 
illustrate that similar ideas can be used even in scenarios where the cut 
polytope does not provide a complete characterization of the noncontextual 
set.

 The paper is 
organized as follows: in Sec. \ref{sec:comp} we  review the definition of 
a compatibility scenario and of noncontextuality in the tradional sense; 
In Sec. \ref{sec:ext}, we review the definition of extended 
noncontextuality of Ref. \cite{KDL15}, stating it in graph-theoretical 
terms; In Sec. \ref{sec:coup} we maximize the probability of  two 
realizations of the same random variables  in different contexts being 
equal; In Sec.~\ref{sec: twoout} focusing on scenarios with two outcomes 
per measurement, we introduce the cut polytope and the extended 
compatibility hypergraph for a scenario and show
a complete characterization of the extended contextuality for the 
$n-$cycle scenario; In Sec.~\ref{sec:valid_ineq} using the introduced cut 
polytope we provide necessary conditions for the existence of 
noncontextual behaviours in any given scenario, although the complete 
characterization is an extremely difficult problem; In 
Sec.~\ref{sec:I3322} and Sec.~\ref{sec:chain_ineq} we apply our methods
for important families of contextuality inequalities; We discuss scenarios 
with more than three measurements in Sec.~\ref{sec:more_than_3} and  
close this work with a discussion in Sec.~\ref{sec:discussion}.

%--------------------------------------------
% Compatibility Scenarios
%--------------------------------------------

\section{Compatibility scenarios}
\label{sec:comp}

\begin{dfn}
 A \emph{compatibility scenario} is defined by a triple $\Gamma :=\left(X, 
\mathcal{C}, O \right)$, where $O$ is a finite set, $X$ is a finite set of 
random variables taking values in $O$, %in  $\left(O, \mathcal{P}\left(O\right)\right)$ 
%\cris{why $O$ and $\mathcal{P}\left(O\right)$??},
and 
  $\mathcal{C}$ is a family of subsets  of $X$ such that
\begin{enumerate}
 \item $\displaystyle \cup_{C \in \mathcal{C}} C =X$;
\item   $C,C' \in \mathcal{C}$ and $C \subseteq C'$
implies $C = C'$. \label {antichain}
\end{enumerate}
The elements $C \in \mathcal{C}$ are called \emph{contexts} and the
set $\mathcal{C}$ is called the \emph{compatibility cover} of the scenario.
\end{dfn}

One may think of the random variables in $X$ as
representing measurements in a physical system, with 
possible outcomes labeled by the elements in $O$, 
while the sets in $\mathcal{C}$ may be thought as
encoding the compatibility relations among the 
measurements in $X$, that is, each set $C \in \mathcal{C}$ consists 
of a maximal set of compatible, jointly measurable random variables 
\cite{AB11, AT17book}. 
Equivalentely, the compatibility relations among 
the elements of $X$ can be represented by an hypergraph. 
%\cris{Hypergraphs in which inside the edges there are other edges do not 
%satisfy \ref{antichain} above, right? =/.}

\begin{dfn}
The \emph{compatibility hypergraph} of a scenario $\left(X, \mathcal{C}, O \right)$ is an hypergraph $\mathrm{H} = \left(X, \mathcal{C}\right)$
whose vertices are 
the random variables in $X$ and  hyperedges are the contexts $C \in \mathcal{C}$.
The \emph{compatibility graph} of the scenario is  the 2-section of $\mathrm{H}$, that is, the graph $\mathrm{G}$ has the same vertices of the hypergraph $\mathrm{H}$ and edges between all pairs of vertices contained in the some hyperedge of $\mathrm{H}$.
\end{dfn}

In an experiment, characterized by a compatibility scenario $\Gamma 
= (X,\mc{C},O)$, when compatible measurements, represented by the random 
variables belonging to a context $C=\{x_1,x_2,...,x_{\vert C \vert } \} 
\in \mc{C}$, are performed jointly, a list $s=(a_1,a_2,...,a_{\vert C 
\vert })$ of outcomes in the Cartesian product 
\begin{equation}
O^{C}:=\underbrace{O \times O \times ... \times O}_{\vert C \vert 
-\mbox{times}}
\label{eq.cartesianproduct}
\end{equation}
is observed. Moreover, the collection of 
well-defined joint probability distributions for the random variables 
associated with $C \in \mc{C}$ receives special attention:

% \st{When the random variables in $C$ are performed jointly, a set of 
% outcomes in $O^C$ will be observed, where $O^C$ denotes the Cartesian 
% product of $|C|$ copies of $O$.} \cris{Ficou bom?} 

\begin{dfn}
A \emph{behavior} $\mathrm{B}$ for the scenario $\left(X, \mathcal{C}, O \right)$ is a family of probability distributions over $O^C$, 
one for each context
$C \in \mathcal{C}$, that is, 
\be \mathrm{B} = \left\{p_C:  O^C\rightarrow [0,1] \left|\sum_{s\in O^C} p_C(s)=1, C \in \mathcal{C}\right.\right\}.\ee
\end{dfn}

This means that for each context $C$, $p_C(s)$ gives the 
probability of obtaining outcomes $s$ in a joint measurement of the 
elements of $C$. Following standard notation in the community,
 given a context $C=\left\{x_1, \ldots , x_{|C|}\right\}$ 
and $s=\left(a_1, \ldots , a_{|C|}\right)$ a  particular list of 
outcomes for those measurements in $C$, we will from now on represent 
$p_C(s)$ as
\be p\left(a_1, \ldots , a_{|C|}\left|x_1, \ldots , x_{|C|}\right.\right).\ee

%Given the state %(\textcolor{red}{meio repentino... não?})
%of the physical 
%system which is being measured, to each context $c_i=\left\{ x_{i1}, x_{i2} 
%\ldots , x_{i\left|c_i\right|}\right\},$ 
%corresponds a probability distribution
%\be p\left(a_{i1}a_{i2} \ldots a_{i\left|c_i\right|}\left| x_{i1} x_{i2} \ldots  x_{i\left|c_i\right|} \right.\right)\ee
%which gives the probability of obtaining outcomes $a_{i1}a_{i2} \ldots a_{i\left|c_i\right|}$
%given that the measurements $x_{i1} x_{i2} \ldots  x_{i\left|c_i\right|} $ 
%were 
%performed jointly. We call a set of probability distributions 
%\be \mathrm{B}_{\mathrm{H}}=\left\{ p\left(a_{i1}a_{i2} \ldots a_{i\left|c_i\right|}\left| x_{i1} x_{i2} \ldots  x_{i\left|c_i\right|} \right.\right)\right\}, \ 
% \\ c_i \in \mathrm{C} \ee a 
%\emph{behavior} for the  hypergraph 
%$\mathrm{H}$. Therefore, each behavior $\mathrm{B}_{\mathrm{H}}$ encodes 
%all the experimental information for the contextualtiy scenario dictated 
%by $\mathrm{H}$, \textit{i.e.} the probability distributions for all 
%possible contexts $c_i$ belonging to $\mathrm{C}$.
    
\textbf{Remark:} Despite of being absolutely standard using the 
above notation for representing an element $p_C$ in a behaviour $B$, to 
avoid misunderstanding within the mathematical community, and to 
make our work more readable for those from other communities who might 
become interested in this topic, we note that the 
mathematical object we are using here is the joint probability
$\mathds{P}(x_1=a_1,x_2=a_2,...,x_{\vert C \vert}=a_{\vert C \vert})$, 
defined on the finite set $O^{C}$.

In an ideal situation, one generally assumes that behaviors are non-disturbing.

\begin{dfn}
\label{definondisturbance}
The \emph{non-disturbance} set $\mathcal{X}\left(\Gamma\right)$ of a 
compatibility scenario $\Gamma$ is the set of behaviors that satisfy
 the consistency relation  
\be\label{eqnondisturbing}
\sum_{a^i_k| x^i_{k} \notin C_i \cap C_j} p\left(a^i_{1}a^i_{2} \ldots  
a^i_{\left|C_i\right|}\left| x^i_{1} x^i_{2} \ldots  x^i_{\left|C_i\right|} \right.\right) = 
\sum_{a^j_{l}|x^j_{l} \notin C_i \cap C_j}  p\left(a^j_{1}a^j_{2} \ldots 
a^j_{\left|C_j\right|}\left| x^j_{1} x^j_{2} \ldots  x^j_{\left|C_j\right|} \right.\right) 
%p(a^{1} a^{2} \ldots a^{|C_i \cap C_j |} \left| C_i \cap C_j 
%\right.) = & \nonumber  
\ee
 for any two intersecting contexts $C_i$ and $C_j$ in  
$\mathcal{C}$, when considering at both sides the same sets of 
outcomes for those measurements in $C_i \cap C_j$.
 \end{dfn}
\textbf{Remark:} Eq.~\eqref{eqnondisturbing} above says that 
when the non-disturbance relation is satisfied in those contexts which 
share some common random variables, it does not matter the way one takes 
the marginalization to these variables into account. Both 
marginalizations, either starting from $C_i$ or starting from $C_j$, must
coincide. 
 
In an hypothetical situation where all measurements
in $\mathrm{X}$ are compatible, it would be possible to define a
\emph{global probability distribution} $p(a_1 a_2 ... a_{\vert X 
\vert} \vert x_1 x_2 ... x_{\vert X \vert})$, or
\be \label{eqglobal}
p\left(a_1a_2 \ldots a_{|X|}\right)
\ee
for short, that would give the probability of obtaining outcomes 
$a_1a_2 \ldots a_{|X|}$ 
as though all measurements in $X$ were 
jointly performed.

\begin{dfn}
A behavior 
$\mathrm{B} \subset \mathcal{X}\left(\Gamma\right)$ is \emph{noncontextual} if there is a global probability distribution \eqref{eqglobal} such that for each $C \in \mathcal{C}$
\be\label{equsualnoncontex} p\left(a_{1}a_{2} \ldots a_{\left|C\right|}\left| x_{1} x_{2} \ldots  x_{\left|C\right|} \right.\right) =
 \sum_ {a_l| l  \notin C}  p\left(a_1a_2 \ldots a_{n}\right),  
\ee
where the sum is taken over the outcomes $a_l$  of the measurements $l 
\notin C$ and $a_l = a_{k}$ for each $l=x_{k} \in C$. 
\end{dfn}

In other words, $\mathrm{B}$ is 
noncontextual if   the probability distribution assigned by $\mathrm{B}$ to each context can be recovered as marginal from the global probability distribution 
$p\left(a_1a_2 \ldots a_{n}\right)$ \cite{Fine82, AB11}.

\section{Extended Contextuality}
\label{sec:ext}

To define noncontextuality in a scenario where the non-disturbance property 
 \eqref{eqnondisturbing} is not valid, we first must change 
the definition of noncontextual behaviors given by Eq. \eqref{equsualnoncontex}. We will consider \emph{extended global 
probability distributions} of the form 
\be \label{eqextendedglobal}
p\left(\underbrace{a^1_{1} \ldots a^1_{\left|C_1\right|}}_{C_1} \underbrace{a^2_{1} \ldots a^2_{\left|C_2\right|}}_{C_2} \ldots
\underbrace{a^m_{1} \ldots a^m_{\left|C_m\right|}}_{C_m}\left| \underbrace{x^1_{1}  \ldots  x^1_{\left|C_1\right|}}_{C_1}
\underbrace{x^2_{1}  \ldots  x^2_{\left|C_2\right|}}_{C_2}
\ldots \underbrace{x^m_{1}  \ldots  x^m_{\left|C_m\right|}}_{C_m}\right.\right),\ee
where $m = \left|\mathcal{C}\right|$, that gives joint probability of 
obtaining outcomes $a^i_{1}, \ldots ,a^i_{\left|C_i\right|}$ for each 
context $ C_i=\left\{x^i_{1},  \ldots ,  x^i_{\left|C_i\right|}\right\}.$
Notice that this extended global probability distribution is, in general, 
not equal to the probability distribution defined in Eq.
\eqref{eqglobal}, since the same random variable could appear in more than 
one 
context, and hence, in the list 
\be \underbrace{x^1_{1}  \ldots  x^1_{\left|C_1\right|}}_{C_1}
\underbrace{x^2_{1}  \ldots  x^2_{\left|C_2\right|}}_{C_2}
\ldots \underbrace{x^m_{1}  \ldots  x^m_{\left|C_m\right|}}_{C_m} \ee \vspace{1em}
the same random variable would be repeated several times.

To make  definitions  in Eqs.\eqref{eqglobal} and 
\eqref{eqextendedglobal} equivalent in the case of non-disturbing behaviors, we demand that, if in  different 
contexts $C_{i_1}, C_{i_2}, \ldots , C_{i_l}$ there exist coincident random variables $x^{i_1}_{k_1}, x^{i_2}_{k_2}, \ldots , x^{i_l}_{k_l}$ ,
then 

\begin{align}
  p\left(a^{i_1}_{k_1}\ldots  a^{i_l}_{k_l}\left|x^{i_1}_{k_1}\ldots  x^{i_l}_{k_l}\right.\right) &  = \nonumber \\ 
 & \sum_{ a^r_{s} \left| (r,s) \neq \left(i_j, k_j\right) \right.} p\left(a^1_{1} \ldots a^1_{\left|C_1\right|} \ldots
a^m_{1} \ldots a^m_{\left|C_m\right|}\left| x^1_{1}  \ldots  x^1_{\left|C_1\right|}
\ldots x^m_{1}  \ldots  x^m_{\left|C_m\right|}\right.\right)  \label{eqcorrelated} \\ &= \left\{\begin{array}{cc}
                                                                         1& \mbox{if } \  a^{i_1}_{k_1}=a^{i_2}_{k_2}=\ldots = a^{i_l}_{k_l}\label{eqcorrelated2} \\
                                                                         0& \mbox{ otherwise} ,
                                                                        \end{array}\right.  
\end{align}
that is, marginal probability distributions  for  $x^{i_1}_{k_1},x^{i_2}_{k_2}, \ldots , x^{i_l}_{k_l}$, representing the same random variable in different contexts,  are perfectly correlated.
Hence, it is equivalent to  say that  $\mathrm{B}$ is a \textit{ noncontextual behavior}
 if there is a 
extended global probability distribution  satisfying condition   \eqref{eqcorrelated2}
 such  that
\be  p\left(a^i_{1}a^i_{2} \ldots a^i_{\left|C_i\right|}\left| x^i_{1} x^i_{2} \ldots  x^i_{\left|c_i\right|} \right.\right)
= \sum_{a^j_{k} | j \neq i}p\left(a^1_{1} 
\ldots a^1_{\left|C_1\right|} \ldots
a^m_{1} \ldots a^m_{\left|c_m\right|}\left| x^1_{1}  \ldots  x^1_{\left|C_1\right|}
\ldots x^m_{1}  \ldots  x^m_{\left|C_m\right|}\right.\right) \label{eqmarginalext}\ee

A simple example of this situation is shown in Fig. \ref{figc3}. 
There, a simple compatibility scenario with three measurements 
$0,1,2$ and two contexts, $\left\{0,1\right\}$ and $\left\{1,2\right\}$
is shown.
 A behaviour for such a scenario consists of two probability distributions 
$p(ab|01)$ and $p(bc|12)$. Traditionally, one says that a non-disturbing 
behavior for 
 this scenario is noncontextual if there is a global probability 
distribution $p(abc)$ such that $p(ab|01)= \sum_c p(abc)$
 and $p(bc|12)=\sum_a p(abc)$. For our purposes it will be convenient to 
consider an extended global probability distribution
 $p(abb'c|011^{\prime}2)$ such that $p(bb'|11^{\prime}) = \sum_{a,c} 
p(abb'c|0112)=1$ iff $b=b'$, and zero otherwise. Then, in this situation, 
we say that a behavior  is noncontextual if
 there is an extended global probability distribution satisfying this 
condition such that $p(ab|01)= \sum_{b',c} p(abb'c|011^{\prime}2)$
 and $p(b'c|12)=\sum_{a,b} p(abb'c|011^{\prime}2)$. For non-disturbing 
behaviors, these two notions
 of noncontextualtiy are equivalent.

\begin{figure}[h!]
\definecolor{ffwwqq}{rgb}{1,0.4,0}
\definecolor{aqaqaq}{rgb}{1,0.6,0.2}
\begin{tikzpicture}[line cap=round,line join=round,>=triangle 45,x=1cm,y=1cm]
\draw [line width=2pt] (-5,1)-- (-2,5);
\draw [line width=2pt] (-2,5)-- (1,1);
\draw [fill=aqaqaq] (-5,1) circle (10pt);
\draw[color=black] (-5,1) node {$0$};
\draw [fill=ffwwqq] (-2,5) circle (10pt);
\draw[color=black] (-2,5) node {$1$};
\draw [fill=aqaqaq] (1,1) circle (10pt);
\draw[color=black] (1,1) node {$2$};
\end{tikzpicture}
 \caption{\footnotesize{A simple compatibility scenario with three 
measurements $0,1,2$ 
and two contexts, $\left\{0,1\right\}$ and 
$\left\{1,2\right\}$. Here the compatibility 
hypergraph associated with the scenario already coincides with its 
compatibility graph.}}
%  A behavior for this scenario consists of two probability distributions 
% $p(ab|01)$ and $p(bc|12)$. Traditionally, one says that a behavior in 
%  this scenario is noncontextual if there is a global probability 
% distribution $p(abc)$ such that $p(ab|01)= \sum_c p(abc)$
%  and $p(bc|12)=\sum_a p(abc)$. For our purposes it will be convenient to 
% consider an extended global probability distribution
%  $p(abb'c|0112)$ such that $p(bb'|11) = \sum_{a,c} p(abb'c|0112)=1$ iff 
% $b=b'$. Then we say that a behavior  is noncontextual if
%  there is an extended global probability distribution satisfying this 
% condition such that $p(ab|01)= \sum_{b',c} p(abb'c|0112)$
%  and $p(b'c|12)=\sum_{a,b} p(abb'c|0112)$. For non-disturbing behaviors, 
% these two notions
%  of noncontextualtiy are equivalent.
 \label{figc3}
\end{figure}
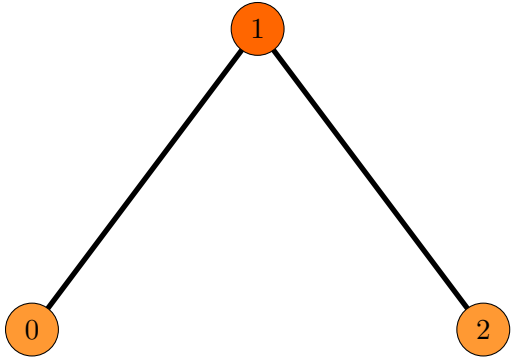

To define noncontextuality in a scenario where the non-disturbance property does not hold, we adopt the strategy of Ref. \cite{KDL15}.
We  relax 
 the requirement that marginals for $x^{i_1}_{k_1},x^{i_2}_{k_2},\ldots ,x^{i_l}_{k_l}$ be
perfectly correlated when they represent the same random variable . Instead of Eq. \eqref{eqcorrelated2}, we require that the probability of 
$x^{i_1}_{k_1},x^{i_2}_{k_2},\ldots , x^{i_l}_{k_l}$
being equal is the maximum allowed by the individual probability distributions of each $x^{i_l}_{k_l}$.

\begin{dfn}
 We say that a behavior has a \emph{maximally noncontextual description} if there is an extended global distribution \eqref{eqextendedglobal} such that 
the distribution of each context is obtained as a marginal, according to Eq.  \eqref{eqmarginalext}, and such that if  
$x^{i_1}_{k_1},x^{i_2}_{k_2},\ldots , x^{i_l}_{k_l}$
represent the same random variable, the marginals for $x^{i_1}_{k_1},x^{i_2}_{k_2}, \ldots , x^{i_l}_{k_l}$ defined by 
Eq. \eqref{eqcorrelated} are such that 
\be p\left(x^{i_1}_{k_1}=\ldots = x^{i_l}_{k_l}\right)=\sum_{a} p\left(a\ldots  a\left|x^{i_1}_{k_1}\ldots  x^{i_l}_{k_l}\right.\right) \ee
is the maximum consistent with the  marginal distributions $p\left(a^{i_j}_{k_j}\left|x^{i_j}_{k_j}\right.\right)$.
That is, a behavior is noncontextual in the extended sense if there is an extended global distribution that gives the correct 
marginal in each context and that maximizes the probability of $x^{i_1}_{k_1},x^{i_2}_{k_2},\ldots , x^{i_l}_{k_l}$ being equal if they 
represent the same random variable in different contexts.
\label{def:ext_context}
\end{dfn}

% \st{Following Ref. %\cite{KDL15},
% when 
% $x_{i_1}^{k_1},x_{i_2}^{k_2},\ldots, 
% x_{i_l}^{k_l}$represent the same measurement, we call a distribution 
% $ p\left(a_{i_1}^{k_1}a_{i_2}^{k_2}\ldots  a_{i_l}^{k_l}\left|x_{i_1}^{k_1}x_{i_2}^{k_2}\ldots  x_{i_l}^{k_l}\right.\right) $
% that gives the correct marginals 
% $p\left(a_{i_j}^{k_j}\left|x_{i_j}^{k_j}\right.\right)$ a \emph{coupling} 
% for 
% $x_{i_1}^{k_1},x_{i_2}^{k_2},\ldots , x_{i_l}^{k_l}$ We say that such a 
% coupling is \emph{maximal}
% if $p\left(x_{i_1}^{k_1}=x_{i_2}^{k_2}=\ldots = x_{i_l}^{k_l}\right)$ achieves the maximum value consistent with the marginals
% $p\left(a_{i_j}^{k_j}\left|x_{i_j}^{k_j}\right.\right)$}.

According Ref~\cite{KDL15} we define \emph{maximal coupling} as 
follows:
\begin{dfn}
Given 
$\{ x_{i_1}^{k_1},x_{i_2}^{k_2},\ldots, 
x_{i_l}^{k_l}\}$ a set of random variables representing the same 
measurement, we call a distribution 
$ p\left(a_{i_1}^{k_1}a_{i_2}^{k_2}\ldots  
a_{i_l}^{k_l}\left|x_{i_1}^{k_1}x_{i_2}^{k_2}\ldots  
x_{i_l}^{k_l}\right.\right) $
that gives the correct marginals 
$p\left(a_{i_j}^{k_j}\left|x_{i_j}^{k_j}\right.\right)$ a \emph{coupling} 
for 
$x_{i_1}^{k_1},x_{i_2}^{k_2},\ldots , x_{i_l}^{k_l}$. We say that such a 
coupling is \emph{maximal}
if $p\left(x_{i_1}^{k_1}=x_{i_2}^{k_2}=\ldots = x_{i_l}^{k_l}\right)$ 
achieves the maximum value consistent with the marginals
$p\left(a_{i_j}^{k_j}\left|x_{i_j}^{k_j}\right.\right)$.
 \label{def:max_coupling}
\end{dfn}

\section{Existence of Maximal Couplings}
\label{sec:coup}

It could be the case that a maximal coupling, as in 
Def.~\ref{def:max_coupling} did not exist for a given set of random 
variables which represents the same measurement. It turns out that it 
would never happen. Here we constructively show that a maximal coupling is 
a well-defined notion. \emph{i.e.} under certain assumptions there always 
exists at least one maximal coupling for a given set of random variables.

\begin{thm}
\label{teomaxcoupling}
 Given a set of random variables $x_{i_1}^{k_1},x_{i_2}^{k_2},\ldots , x_{i_l}^{k_l}$ with 
 distributions $p\left(a_{i_j}^{k_j}\left|x_{i_j}^{k_j}\right.\right)$ it is always possible to construct  a maximal 
 coupling for  this set with
 \be 
 p\left(x_{i_1}^{k_1}=x_{i_2}^{k_2}=\ldots =
x_{i_l}^{k_l}\right)= \sum_{a} \min_{j}\left\{ 
p\left(a\left|x_{i_j}^{k_j}\right.\right)\right\}.
 \ee
\end{thm}
\begin{proof}
Let \be p_-(a)= \min_{j} \left\{ p\left(a\left|x_{i_j}^{k_j}\right.\right)\right\}.\ee Then
 \be p\left(x_{i_1}^{k_1}=x_{i_2}^{k_2}=\ldots = x_{i_l}^{k_l}=a\right) \leq p_-(a) \label{eqp=max}\ee
 and hence,
 \be p\left(x_{i_1}^{k_1}=\ldots = x_{i_l}^{k_l}\right)  = 
 \sum_a  p\left(x_{i_1}^{k_1}=\ldots = x_{i_l}^{k_l}=a\right)
 \leq  \sum_a  p_-(a). \ee
 Construct the coupling as follows: if 
$a_{i_1}^{k_1} = a_{i_2}^{k_2} = \ldots =  a_{i_l}^{k_l}=a$, we 
define
 \be p\left(a\ldots a \left| x_{i_1}^{k_1}x_{i_2}^{k_2}\ldots  x_{i_l}^{k_l}\right.\right)=p_-(a);\ee
 if not, 
 we define $$p\left(a_{i_1}^{k_1}a_{i_2}^{k_2}\ldots  a_{i_l}^{k_l}\left|x_{i_1}^{k_1}x_{i_2}^{k_2}\ldots  x_{i_l}^{k_l}\right.\right)= 
 \prod_j p'\left(a_{i_j}^{k_j}\left| x_{i_j}^{k_j}\right.\right),$$
 where 
\be p'\left(a_{i_j}^{k_j}\left|x_{i_j}^{k_j}\right.\right)= p\left(a_{i_j}^{k_j}\left|x_{i_j}^{k_j}\right.\right)-p_-(a).\ee
 This defines a maximal coupling for $x_{i_1}^{k_1},x_{i_2}^{k_2},\ldots , x_{i_l}^{k_l}$.
 \end{proof}

One should notice that although the method we have applied in the 
proof above provides a maximal coupling for the considered set of 
random variables, there is no guarantee that such a coupling is the unique 
consistent with Def.~\ref{def:max_coupling} when treating with the general 
case. Actually, it turns out that in some specific 
situations the coupling constructed above is indeed unique. This is 
always the case, for example, for two variables with any
number of outcomes and three variables each of which with two 
outcomes.
%\st{It is important to notice, however, that this is not generally true.} 

\begin{thm}[Sufficient condition for extended contextualtiy]
If there is  an extended global distribution, as in Eq. 
\eqref{eqextendedglobal}, such that the marginals in each context are 
equal to the distributions of the  behavior $\mathrm{B}$, according to Eq. 
 \eqref{eqmarginalext}, and such that the corresponding couplings for each 
set 
$x_{i_1}^{k_1},x_{i_2}^{k_2},\ldots , x_{i_l}^{k_l}$
representing the same random variable, defined by 
Eq. \eqref{eqcorrelated}, are equal to the ones given in  Thm. \ref{teomaxcoupling}, then the behavior $\mathrm{B}$
is noncontextual in the extended sense. 
 \label{teomaxcouplinge}
\end{thm}

The condition stated in Thm. \ref{teomaxcouplinge}  is also necessary when the coupling constructed in 
the proof of Thm. \ref{teomaxcoupling} is unique.

%\begin{cor}
%If each measurement belongs to at most two contexts or has two outcomes and belongs to at most three contexts, the behavior $\mathrm{B}$
%is noncontextual in the extended sense
%if, and only if, there is  an extended global distribution \eqref{eqextendedglobal} such that the marginals in each context are equal to the 
%distributions of the  behavior $\mathrm{B}_{\mathrm{H}}$, according to Equation  \eqref{eqmarginalext}, and such that the corresponding 
%couplings for each set of random variables
%$x_{i_1}^{k_1},x_{i_2}^{k_2},\ldots , x_{i_l}^{k_l}$
%representing the same measurement, defined by 
%Equation \eqref{eqcorrelated}, are equal to the ones given in  Theorem \ref{teomaxcoupling}. 
 
%\end{cor}

When the  coupling   given in  Thm. \ref{teomaxcoupling} is not unique, the 
difference between any other coupling and the 
one constructed in the proof of this theorem can only appear in the terms
\be  p\left(a_{i_1}^{k_1}a_{i_2}^{k_2}\ldots  a_{i_l}^{k_l}\left|x_{i_1}^{k_1}x_{i_2}^{k_2}\ldots  x_{i_l}^{k_l}\right.\right)\ee 
for which the outcomes $a_{i_1}^{k_1}, a_{i_2}^{k_2}, \ldots , a_{i_l}^{k_l}$ are not all equal. Otherwise this would contradict the 
hypotheses that the coupling is maximal. Then for any maximal coupling we
can at least say that for each pair $x_{i_m}^{k_m}$ and $x_{i_n}^{k_n}$
we have
\be
 p_-(a)=  p\left(x_{i_1}^{k_1}=  \ldots = x_{i_l}^{k_l}\right)  \leq 
  p\left(x_{i_m}^{k_m}=x_{i_n}^{k_n}=a\right) \leq 
   \min_{m,n}   \left\{ p\left(a \left| x_{i_m}^{k_m}\right.\right), p\left(a\left| x_{i_n}^{k_n}\right.\right)\right\} .  
\ee
This relation will be used to construct  necessary condition for extended noncontextuality in Sec. \ref{sec:valid_ineq}.

\begin{thm}[Necessary condition for maximal coupling]
\label{thm:necessary_max_coup}
 If $p\left(a_{i_1}^{k_1}a_{i_2}^{k_2}\ldots  a_{i_l}^{k_l}\left|x_{i_1}^{k_1}x_{i_2}^{k_2}\ldots  x_{i_l}^{k_l}\right.\right)$ is a 
 maximal coupling for the random  variables $x_{i_1}^{k_1},x_{i_2}^{k_2},\ldots , x_{i_l}^{k_l}$, then 
 \be
p_{-}(a)=p(x_{i_1}^{k_1},...,x_{i_l}^{k_l}) \leq 
p\left( \left\lbrace x_{i_j}^{k_j}=a \right\rbrace_{j \in \mc{S}} 
\right)\leq \min_{j \in \mc{S}}\left\lbrace p(a \vert x_{i_j}^{k_j}) 
\right\rbrace 
\ee
for any subset $\mc{S} \subset [l]$.
\end{thm}

%  \be
%   p_-(a) \leq   p\left(x_{i_m}^{k_m}=x_{i_n}^{k_n}=a\right)  \leq 
%  \min_{m,n}  \left\{ p\left(a\left| x_{i_m}^{k_m}\right.\right), p\left(a\left| x_{i_n}^{k_n}\right.\right)\right\} . 
%  \ee
%  for all $m,n =1, \ldots, l$.
% \end{thm}
% 
% \textbf{Remark:} One should notice that, given a maximal coupling, a 
% similar reasoning could be applied for any subset $\mc{S} \subset [l]$: 
% \be
% p_{-}(a)=p(x_{i_1}^{k_1},...,x_{i_l}^{k_l}) \leq 
% p\left( \left\lbrace x_{i_j}^{k_j}=a \right\rbrace_{j \in \mc{S}} 
% \right)\leq \min_{j \in \mc{S}}\left\lbrace p(a \vert x_{i_j}^{k_j}) 
% \right\rbrace 
% \ee

\section{Two outcomes}
\label{sec: twoout}

When $O=\left\{ -1,1\right\} $ we can use a powerful tool from graph theory to find 
necessary conditions for noncontextuality: the \emph{cut 
polytope}~\cite{AIT06,DL97,Amaral14}.

\begin{dfn}[Cut Polytope]
 The \emph{cut polytope} of a graph $G=(V,E)$, 
denoted by $CUT(G)$, is the convex hull of the set $\mc{V}=\left\lbrace 
\delta_{G}(S) \in \mathds{R}^{E}; S \subset V \right\rbrace$ which contains 
all cut vectors of $G$. Given $S \subset V$, the cut vector $\delta_{G}(S) 
\in \mathds{R}^{E}$ associated with $S$ is defined as:
\be
\delta_{u,v}:=
\begin{cases}
1, \,\, \mbox{if} \,\, \vert S \cap \{ u,v \} \vert =1 \\
0, \,\, \mbox{otherwise}.
\end{cases}
\ee
\end{dfn}

Let $\mathrm{G}=\left(\mathrm{X}, \mathrm{E}\right)$ be the compatibility 
graph of a scenario $\Gamma=(X,\mc{C},\{-1,1\})$.
Given a behaviour $\mathrm{B}$, let $\mathrm{P}_{\mathrm{B}} \in 
 \mathds{R}^{X} \times \mathds{R}^{E}$ 
be the vector whose first $\left|\mathrm{X}\right|$ entries are the 
expectation values of the random variables in $\mathrm{X}$ 
\be 
P_{x}:= \left\langle x\right\rangle = p\left(1|x\right)-p\left(-1|x\right), 
x \in \mathrm{X} \label{eq:p1}
\ee
and whose $\left|\mathrm{E}\right|$ subsequent entries are the expectation 
values of product of pairs of compatible random variables in $\mathrm{X}$
\be
P_{xy}:=\left\langle xy\right\rangle = p\left(x=y\right)-p\left(x\neq 
y\right), \left(x,y\right) \in \mathrm{E}.\label{eq:p2}
\ee
Let $\nabla \mathrm{G}$ be the \emph{suspension graph} of $\mathrm{G}$,  
obtained from $\mathrm{G}$ by adding one new vertex $u$ to 
$\mathrm{X}$ which is adjacent to all the other vertices (see 
Fig.~\ref{fig:suspension_graph}).

\begin{figure}
\begin{subfigure}[b]{0.45\columnwidth}
\definecolor{zzttqq}{rgb}{0.6,0.2,0}
\definecolor{ccqqqq}{rgb}{0.8,0,0}
\definecolor{ffubub}{rgb}{1,0.29411764705882354,0.29411764705882354}
\begin{tikzpicture}[scale=0.8,line cap=round,line join=round,>=triangle 45,x=1cm,y=1cm]
%\clip(-7.72,-8.46) rectangle (7.74,3.9);
\draw [line width=2pt,color=black] (-5,-5) -- (-2,-5) -- (-0.12953059442379988,-2.654505552595911) -- (-0.7970933962927425,0.2702781839495594) -- (-3.5,1.571929401302234) -- (-6.202906603707257,0.27027818394956027) -- (-6.8704694055762,-2.65450555259591) -- cycle;
% \draw [line width=2pt] (-5,-5)-- (-2,-5);
% \draw [line width=2pt] (-2,-5)-- (-0.12953059442379988,-2.654505552595911);
% \draw [line width=2pt] (-0.12953059442379988,-2.654505552595911)-- (-0.7970933962927425,0.2702781839495594);
% \draw [line width=2pt] (-0.7970933962927425,0.2702781839495594)-- (-3.5,1.571929401302234);
% \draw [line width=2pt] (-3.5,1.571929401302234)-- (-6.202906603707257,0.27027818394956027);
% \draw [line width=2pt] (-6.202906603707257,0.27027818394956027)-- (-6.8704694055762,-2.65450555259591);
% \draw [line width=2pt] (-6.8704694055762,-2.65450555259591)-- (-5,-5);
% \draw [line width=2pt] (-3.5,-1.8852179051414955)-- (-3.5,1.571929401302234);
% \draw [line width=2pt] (-3.5,-1.8852179051414955)-- (-6.202906603707257,0.27027818394956027);
% \draw [line width=2pt] (-3.5,-1.8852179051414955)-- (-6.8704694055762,-2.65450555259591);
% \draw [line width=2pt] (-3.5,-1.8852179051414955)-- (-5,-5);
% \draw [line width=2pt] (-3.5,-1.8852179051414955)-- (-2,-5);
% \draw [line width=2pt] (-3.5,-1.8852179051414955)-- (-0.12953059442379988,-2.654505552595911);
% \draw [line width=2pt] (-3.5,-1.8852179051414955)-- (-0.7970933962927425,0.2702781839495594);
\draw [fill=ccqqqq] (-5,-5) circle (10pt);
\draw [fill=ccqqqq] (-2,-5) circle (10pt);
\draw [fill=ccqqqq] (-0.12953059442379988,-2.654505552595911) circle (10pt);
\draw [fill=ccqqqq] (-0.7970933962927425,0.2702781839495594) circle (10pt);
\draw [fill=ccqqqq] (-3.5,1.571929401302234) circle (10pt);
\draw [fill=ccqqqq] (-6.202906603707257,0.27027818394956027) circle (10pt);
\draw [fill=ccqqqq] (-6.8704694055762,-2.65450555259591) circle (10pt);
% \draw [fill=ffubub] (-3.5,-1.8852179051414955) circle (10pt);
% \draw[color=black] (-3.5,-1.8852179051414955) node {$u$};
\end{tikzpicture}
\caption{}
\end{subfigure}
\begin{subfigure}[b]{0.45\columnwidth}
 \definecolor{zzttqq}{rgb}{0.6,0.2,0}
\definecolor{ccqqqq}{rgb}{0.8,0,0}
\definecolor{ffubub}{rgb}{1,0.29411764705882354,0.29411764705882354}
\begin{tikzpicture}[scale=0.8,line cap=round,line join=round,>=triangle 45,x=1cm,y=1cm]
%\clip(-7.72,-8.46) rectangle (7.74,3.9);
\draw [line width=2pt,color=black] (-5,-5) -- (-2,-5) -- (-0.12953059442379988,-2.654505552595911) -- (-0.7970933962927425,0.2702781839495594) -- (-3.5,1.571929401302234) -- (-6.202906603707257,0.27027818394956027) -- (-6.8704694055762,-2.65450555259591) -- cycle;
% \draw [line width=2pt] (-5,-5)-- (-2,-5);
% \draw [line width=2pt] (-2,-5)-- (-0.12953059442379988,-2.654505552595911);
% \draw [line width=2pt] (-0.12953059442379988,-2.654505552595911)-- (-0.7970933962927425,0.2702781839495594);
% \draw [line width=2pt] (-0.7970933962927425,0.2702781839495594)-- (-3.5,1.571929401302234);
% \draw [line width=2pt] (-3.5,1.571929401302234)-- (-6.202906603707257,0.27027818394956027);
% \draw [line width=2pt] (-6.202906603707257,0.27027818394956027)-- (-6.8704694055762,-2.65450555259591);
% \draw [line width=2pt] (-6.8704694055762,-2.65450555259591)-- (-5,-5);
\draw [line width=2pt] (-3.5,-1.8852179051414955)-- (-3.5,1.571929401302234);
\draw [line width=2pt] (-3.5,-1.8852179051414955)-- (-6.202906603707257,0.27027818394956027);
\draw [line width=2pt] (-3.5,-1.8852179051414955)-- (-6.8704694055762,-2.65450555259591);
\draw [line width=2pt] (-3.5,-1.8852179051414955)-- (-5,-5);
\draw [line width=2pt] (-3.5,-1.8852179051414955)-- (-2,-5);
\draw [line width=2pt] (-3.5,-1.8852179051414955)-- (-0.12953059442379988,-2.654505552595911);
\draw [line width=2pt] (-3.5,-1.8852179051414955)-- (-0.7970933962927425,0.2702781839495594);
\draw [fill=ccqqqq] (-5,-5) circle (10pt);
\draw [fill=ccqqqq] (-2,-5) circle (10pt);
\draw [fill=ccqqqq] (-0.12953059442379988,-2.654505552595911) circle (10pt);
\draw [fill=ccqqqq] (-0.7970933962927425,0.2702781839495594) circle (10pt);
\draw [fill=ccqqqq] (-3.5,1.571929401302234) circle (10pt);
\draw [fill=ccqqqq] (-6.202906603707257,0.27027818394956027) circle (10pt);
\draw [fill=ccqqqq] (-6.8704694055762,-2.65450555259591) circle (10pt);
\draw [fill=ffubub] (-3.5,-1.8852179051414955) circle (10pt);
\draw[color=black] (-3.5,-1.8852179051414955) node {$u$};
\end{tikzpicture}
\caption{}
\end{subfigure}
 \caption{\footnotesize{ An example of a suspension graph. On the right 
hand side we 
depicted a 7-cycle, whereas on the left it is depicted its 
suspension graph, with a new vertex $u$ added to the vertex set, and
connected to each other vertex already belonging to the 7-cycle.}}
 \label{fig:suspension_graph}
\end{figure}
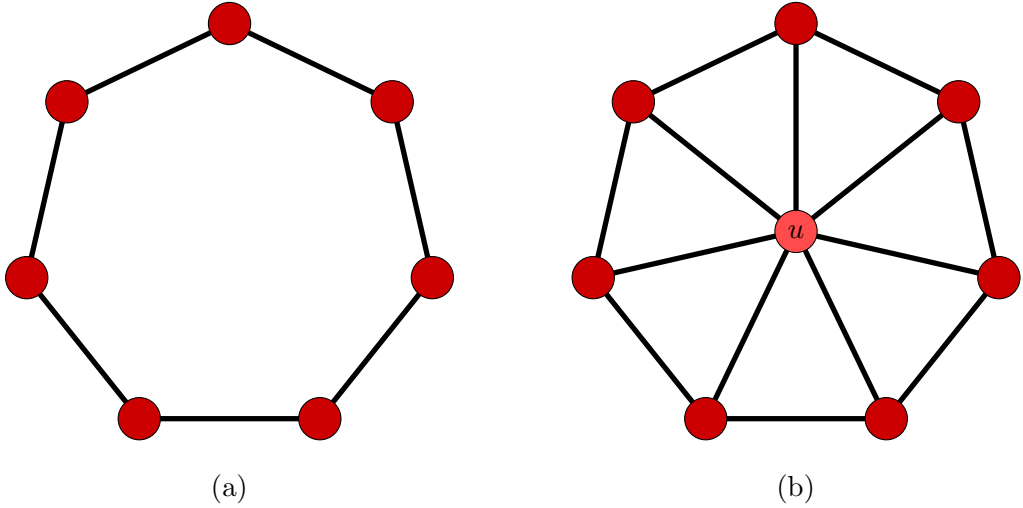

\begin{prop}
Let $\Gamma=(X,\mc{C},\{-1,1\})$ be a compatibility scenario, and 
let $G$ be the compatibility graph associated with $\Gamma$. If a 
behavior $\mathrm{B}$ is noncontextual,
 the vector $\mathrm{P}_{\mathrm{B}}$   belongs to the cut polytope of $\nabla \mathrm{G}$.
 \label{thm:nec_cut}
\end{prop}

For a proof of this result, see Refs. \cite{AII06,AT17,AT17book}. It implies that
characterizing completely the cut polytope $\mathrm{CUT}\left(\nabla \mathrm{G}\right)$ gives a strong necessary condition for 
noncontextuality.
However, as shown in references \cite{Pitowsky91, DL97, AII06, AII08}, such a
characterization is unlikely, since membership testing in
this polytope is a NP-complete
problem~\cite{LexSch03} for general graphs. To do so requires one to find 
all linear inequalities that define the facets of $\mathrm{CUT}\left(\nabla 
\mathrm{G}\right)$,
which is only feasible for limited, although important, scenarios. Nevertheless, one can generally find necessary conditions for
membership in $\mathrm{CUT}\left(\nabla\mathrm{ G}\right)$,
which can be used  to witness contextuality in scenarios where a complete characterization of $\mathrm{CUT}\left(\nabla \mathrm{G}\right)$ is
still missing. 

\begin{dfn}
 Given  $\mathrm{A} \in \mathbb{R}^{\left|\mathrm{X}\right|+\left|\mathrm{E}\right|}$ and $ \mathrm{b} \in \mathbb{R}$ we say that the linear inequality 
 \be \mathrm{A} \cdot \mathrm{P} \leq \mathrm{b} , \ee
 on $\mathrm{P} \in \mathbb{R}^{\left|\mathrm{X}\right|+\left|\mathrm{E}\right|}$ is a \emph{noncontextuality inequality}  if it
 is satisfied for all $\mathrm{P} \in  \mathrm{CUT}\left(\nabla \mathrm{G}\right)$. We say that this inequality is \emph{tight} if 
 $\mathrm{A} \cdot \mathrm{P} = \mathrm{b}$ for some $\mathrm{P} \in  \mathrm{CUT}\left(\nabla\mathrm{ G}\right)$ and we say that this inequality is
 \emph{facet-defining} if the set \be \left\{ \mathrm{P} \in  \mathrm{CUT}\left(\nabla \mathrm{G}\right) \left| \mathrm{A} \cdot \mathrm{P} = \mathrm{b} \right.\right\}\ee
 is a facet of $\mathrm{CUT}\left(\nabla \mathrm{G}\right)$.
\end{dfn}

Every noncontextuality inequality gives a necessary condition for noncontextuality in the corresponding scenario. What we do next is to
use known inequalities valid for $\mathrm{CUT}\left(\nabla G\right)$ to find necessary conditions for  noncontextuality in the extended sense.

\subsection{Extended compatibility hypergraph}
\label{subsec:Extended_Comp_Hyper}

\begin{dfn}
Let $\mathrm{H}$ be the compatibility hypergraph for a compatibility 
scenario  $\Gamma=(X,\mc{C},\{-1,1\})$. Construct the 
\emph{extended compatibility hypergraph}
$\mathscr{H}$ of this scenario in the following way. 
Given a vertex $x \in X$, let $C_{i_1}, \ldots , C_{i_l}$ be all 
hyperedges containing it. We add to the vertex set of $\mathscr{H}$ the 
vertices 
$x^{i_1}, \ldots , x^{i_l}$, which form a hyperedge in $\mathscr{H}$. The 
other hyperedges of $\mathscr{H}$ are in one-to-one correspondence
with the hyperedges of $\mathrm{H}$: to each hyperedge $ C_i=\left\{x_{1}, 
x_{2}, \ldots ,  x_{\left|C_i\right|}\right\}$ in $\mathrm{H}$ corresponds
the hyperedge $\left\{x_{1}^i, x_{2}^i, \ldots ,  
x_{\left|C_i\right|}^i\right\}$ in $\mathscr{H}$.
\end{dfn}

 Fig.\ref{fig:path_ext} illustrates this construction for 
a simple example.

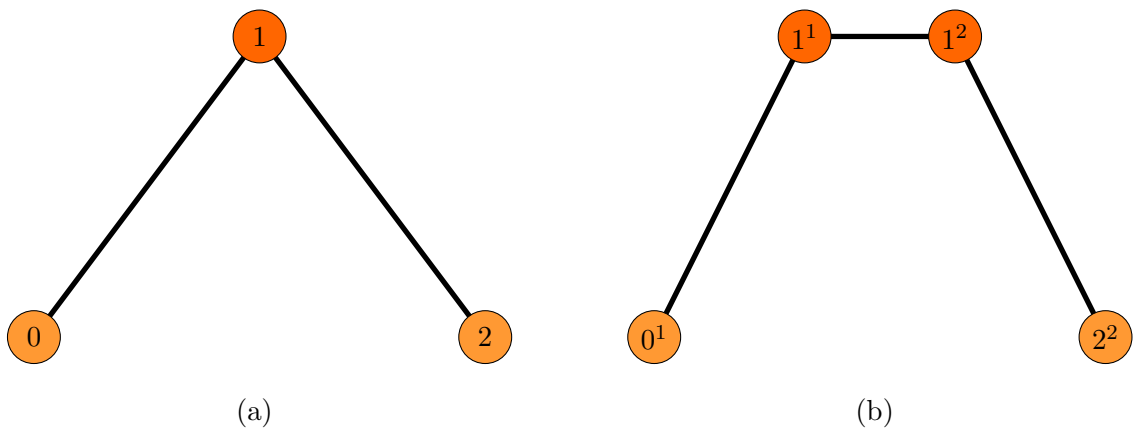
\begin{figure}[h!]
 \begin{subfigure}[b]{0.4\textwidth}
\definecolor{ffwwqq}{rgb}{1,0.4,0}
\definecolor{aqaqaq}{rgb}{1,0.6,0.2}
\begin{tikzpicture}[line cap=round,line join=round,>=triangle 45,x=1cm,y=1cm]
\draw [line width=2pt] (-5,1)-- (-2,5);
\draw [line width=2pt] (-2,5)-- (1,1);
\draw [fill=aqaqaq] (-5,1) circle (10pt);
\draw[color=black] (-5,1) node {$0$};
\draw [fill=ffwwqq] (-2,5) circle (10pt);
\draw[color=black] (-2,5) node {$1$};
\draw [fill=aqaqaq] (1,1) circle (10pt);
\draw[color=black] (1,1) node {$2$};
\end{tikzpicture}
\caption{ }
\end{subfigure}
\quad \qquad \quad
 \begin{subfigure}[b]{0.4\textwidth}
\definecolor{ffwwqq}{rgb}{1,0.4,0}
\definecolor{aqaqaq}{rgb}{1,0.6,0.2}
\begin{tikzpicture}[line cap=round,line join=round,>=triangle 45,x=1cm,y=1cm]
\draw [line width=2pt] (-5,1)-- (-3,5);
\draw [line width=2pt] (-1,5)-- (1,1);
\draw [line width=2pt] (-1,5)-- (-3,5);
\draw [fill=aqaqaq] (-5,1) circle (10pt);
\draw[color=black] (-5,1) node {$0^1$};
\draw [fill=ffwwqq] (-3,5) circle (10pt);
\draw[color=black] (-3,5) node {$1^1$};
\draw [fill=ffwwqq] (-1,5) circle (10pt);
\draw[color=black] (-1,5) node {$1^2$};
\draw [fill=aqaqaq] (1,1) circle (10pt);
\draw[color=black] (1,1) node {$2^2$};
\end{tikzpicture}
\caption{  }
 
 \end{subfigure}
 \caption{\footnotesize{ (a)  Compatibility hypergraph $\mathrm{H}$ of  
scenario with three measurements $0,1,2$ and two contexts, 
$C_1=\left\{0,1\right\}$ and $C_2=\left\{1,2\right\}$.
 (b) The extended compatibility hypergraph $\mathscr{H}$ of $\mathrm{H}$. Measurement $0$ belongs only to context $C_1$, measurement $2$ belongs only
 to context  $C_2$, while measurement $1$ belongs to both contexts. Hence,  the vertex set of $\mathscr{H}$ is
 $\left\{0^1, 1^1, 1^2, 2^2\right\}$. The hyperedges of $\mathscr{H}$ are $\left\{1^1, 1^2\right\}$,  $\left\{0^1, 1^1\right\}$
 and $ \left\{1^2, 2^2\right\}$, the last two being the ones  corresponding 
to those of $\mathrm{H}$.}}
  \label{fig:path_ext}
\end{figure}

\begin{dfn}
Given a behavior $\mathrm{B}$ for the compatibility scenario defined by  hypergraph $\mathrm{H}$,  we construct an \emph{extended behavior} $\mathscr{B}$
for  $\mathrm{B}$ in the following way:
for context $\left\{x_{1}^i x_{2}^i \ldots  x_{\left|C_i\right|}^i\right\}$ of $\mathscr{H}$ corresponding to context $C_i=\left\{x_1 x_2 \ldots  x_{\left|C_i\right|}\right\}$ of $\mathrm{H}$ the probability distribution  assigned by behavior $\mathscr{B}$ is equal to the probability distribution assigned to $C_i$ by behavior 
$\mathrm{B}$;
 for context  $x^{i_1}, \ldots , x^{i_l}$ of $\mathscr{H}$ corresponding to a vertex $x \in X$ of $\mathrm{H}$, the probability distribution  assigned by behavior $\mathscr{B}$ is any maximal coupling for the variables $x^{i_1}, \ldots , x^{i_l}$.
 \end{dfn}
 
 Since, in general, maximal couplings are not unique, for a given 
behaviour $B$, there might exist more than only one extended behaviour 
$\mathscr{B}$ associated with it. In other words, $\mathscr{B}$ will 
also not be unique. With these definitions, we can rewrite Dfn. 
\ref{def:ext_context} as the following theorem:

\begin{thm}
A behavior $\mathrm{B}$ for the compatibility scenario defined by the hypergraph $\mathrm{H}$ has a \emph{maximally noncontextual description} 
 if, and only if, there is an extended behavior  $\mathscr{B}$
for  $\mathrm{B}$ which is noncontextual with respect to the compatibility scenario defined by the extended compatibility hypergraph  $\mathscr{H}$.
\label{thm:extended}
\end{thm}

Thus, the problem of deciding if a behavior $\mathrm{B}$ is noncontextual in the extended sense is equivalent to the problem of 
finding a noncontextual extended behavior $\mathscr{B}$ which is noncontextual in the extended scenario $\mathscr{H}$. This gives, as a corollary,  a complete characterization of extended contextuality for the $n$-cycle scenario.

\subsection{The $n$-cycle scenario}
\label{sub:ncycle}
In the $n$-cycle scenario, $X=\left\{0,\ldots , n-1\right\}$ and two 
measurements $i$ and $j$ are compatible iff $j=i+1 \mod n$. The 
corresponding hypergraph $\mathrm{H}$ is a cycle with $n$ vertices. The 
extended hypergraph $\mathscr{H}$ is a  $2n$-cycle, with vertices 
$i^i, i^{i+1}$ and egdes $\left\{i^i, 
(i+1)^i\right\}, \left\{i^i, i^{i-1} \right\}, \ i=0, \ldots , n-1$ (see 
Fig. \ref{fig:ncycle}).

\begin{figure}[h!]
\begin{subfigure}[b]{0.4\textwidth}
\definecolor{zzccff}{rgb}{0,1,1}
\definecolor{zzccffa}{rgb}{0.2,0.3,1}
\definecolor{zzccffb}{rgb}{0.6,0.8,1}
\definecolor{zzccffc}{rgb}{0,0.6,1}
\definecolor{zzccffd}{rgb}{0,0.5,0.7}
\begin{tikzpicture}[scale=0.7, line cap=round,line join=round,>=triangle 45,x=1cm,y=1cm]
\draw[line width=2pt] (1,0) -- (5,0) -- (6.23606797749979,3.804226065180613) -- (3,6.155367074350506) -- (-0.23606797749978936,3.8042260651806146) -- cycle;
\draw [fill=zzccffb] (1,0) circle (12pt);
\draw[color=black] (1,0) node {$3$};
\draw [fill=zzccffa] (5,0) circle (12pt);
\draw[color=black] (5,0) node {$2$};
\draw [fill=zzccffd] (6.23606797749979,3.804226065180613) circle (12pt);
\draw[color=black] (6.23606797749979,3.804226065180613) node {$1$};
\draw [fill=zzccff] (3,6.155367074350506) circle (12pt);
\draw[color=black] (3,6.155367074350506) node {$0$};
\draw [fill=zzccffc] (-0.23606797749978936,3.8042260651806146) circle (12pt);
\draw[color=black] (-0.23606797749978936,3.8042260651806146) node {$4$};
\end{tikzpicture}
\caption{}
\end{subfigure}
\begin{subfigure}[b]{0.4\textwidth}
\definecolor{zzccff}{rgb}{0,1,1}
\definecolor{zzccffa}{rgb}{0.2,0.3,1}
\definecolor{zzccffb}{rgb}{0.6,0.8,1}
\definecolor{zzccffc}{rgb}{0,0.6,1}
\definecolor{zzccffd}{rgb}{0,0.5,0.7}
\begin{tikzpicture}[scale=0.7, line cap=round,line join=round,>=triangle 45,x=1cm,y=1cm]
\draw[line width=2pt] (16,0) -- (18,0) -- (19.618033988749893,1.1755705045849458) -- (20.23606797749979,3.0776835371752522) -- (19.618033988749893,4.979796569765559) -- (18,6.155367074350505) -- (16,6.155367074350506) -- (14.381966011250107,4.9797965697655595) -- (13.76393202250021,3.0776835371752536) -- (14.381966011250105,1.1755705045849467) -- cycle;
\begin{scriptsize}
\draw [fill=zzccffb] (16,0) circle (12pt);
\draw[color=black] (16,0) node {$3^2$};
\draw [fill=zzccffa] (18,0) circle (12pt);
\draw[color=black] (18,0) node {$2^2$};
\draw [fill=zzccffa] (19.618033988749893,1.1755705045849458) circle (12pt);
\draw[color=black] (19.618033988749893,1.1755705045849458) node {$2^1$};
\draw [fill=zzccffd] (20.23606797749979,3.0776835371752522) circle (12pt);
\draw[color=black] (20.23606797749979,3.0776835371752522) node {$1^1$};
\draw [fill=zzccffd] (19.618033988749893,4.979796569765559) circle (12pt);
\draw[color=black] (19.618033988749893,4.979796569765559) node {$1^0$};
\draw [fill=zzccff] (18,6.155367074350505) circle (12pt);
\draw[color=black] (18,6.155367074350505) node {$0^0$};
\draw [fill=zzccff] (16,6.155367074350506) circle (12pt);
\draw[color=black] (16,6.155367074350506) node {$0^4$};
\draw [fill=zzccffc] (14.381966011250107,4.9797965697655595) circle (12pt);
\draw[color=black] (14.381966011250107,4.9797965697655595) node {$4^4$};
\draw [fill=zzccffc] (13.76393202250021,3.0776835371752536) circle (12pt);
\draw[color=black] (13.76393202250021,3.0776835371752536) node {$4^3$};
\draw [fill=zzccffb] (14.381966011250105,1.1755705045849467) circle (12pt);
\draw[color=black] (14.381966011250105,1.1755705045849467) node {$3^3$};
\end{scriptsize}
\end{tikzpicture}
\caption{}
\end{subfigure}
\caption{\footnotesize{(a) The compatibility hypergraph $\mathrm{H}$ of 
the $5$-cycle 
scenario, which consists of five measurements $0, \ldots , 4$ and five 
contexts $\left\{i,i+1\right\}$, $i=0, \ldots, 4$, the sum being  taken 
$\mod 5$. (b)  The extended compatibility hypergraph $\mathscr{H}$ 
of the $5$-cycle scenario, which is a $10$-cycle with vertices $i^{i-1}, 
i^{i}$ and egdes $\left\{i^i, (i+1)^i\right\}, \left\{i^i, i^{i-1} 
\right\}, \ i=0, \ldots , 4$. }}
\label{fig:ncycle}
\end{figure}
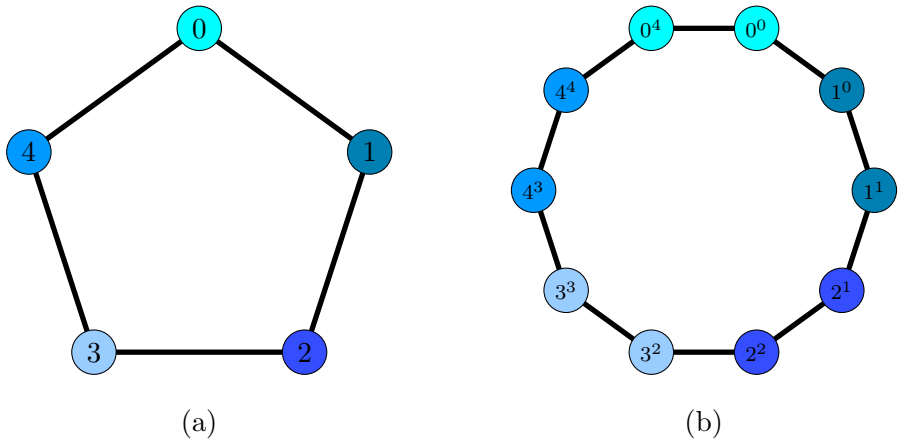

\begin{cor}
A behavior $B$ for the $n$-cycle scenario is noncontextual in the extended sense iff
\be 
s\left(\left\langle {i}^i(i+1)^i\right\rangle, 1-\left\langle 
i^{i}\right\rangle - \left\langle i^{i-1}\right\rangle\right)_{i=0, \ldots 
, n-1} \leq 2n-2,\label{eq:ncycle_ineq}
\ee
where
\be s\left(z_1, \ldots, z_k\right) = \max_{\gamma_i = \pm 1, \prod_i 
\gamma_i=-1 } \sum_{i=1}^k \gamma_i z_i. \label{eq:ncycle}\ee
\end{cor}

\begin{proof}
 In this case the extended behavior $\mathscr{H}$ is unique and, as shown in 
Ref. \cite{KDL15}, for every context $\left\{i^{i-1}, i^{i}\right\}$ corresponding to $i \in X$  we have that maximal couplings satisfy:
\be  \left\langle i^{i-1}i^{i} \right\rangle=1-\left\langle i^{i-1}\right\rangle - \left\langle i^i\right\rangle.\ee
Hence, 
\be P_{\mathscr{B}}= \left(\left\langle i^i(i+1)^i\right\rangle, 1-\left\langle i^{i-1}\right\rangle - \left\langle i^i\right\rangle\right)_{i=0, \ldots , n-1}.\ee
As shown in Ref. \cite{AQBTC13}, 
Eq.~\eqref{eq:ncycle_ineq} is a necessary and sufficient condition 
for membership in $\mathrm{CUT}\left(\nabla\mathrm{C}_{2n}\right).$
Thm. \ref{thm:extended} implies the result.
\end{proof}

\section{From valid inequalities for $\nabla \mathrm{G}$ to valid inequalities for $\nabla \mathscr{G}$}
\label{sec:valid_ineq}

The problem of deciding if a given behavior is noncontextual in the 
extended sense is, in general, extremely difficult (see, for 
instance Ref.~\cite{LexSch03}). To completely solve it we need, first, to 
characterize the set of all extended behaviors $\mathscr{B}$ and, second, 
characterize the set of noncontextual behaviors in the extended scenario, 
which is, as we mentioned before, a complex task. Although we cannot  
solve the problem completely, except for very special situations as in 
Sub.~\ref{sub:ncycle}, we are able to find necessary conditions for the 
existence of a noncontextual extended behavior in \emph{any} scenario using 
the cut polytope. The first step in this direction consists in 
defining a useful and important graph, associated with a given scenario, 
which is going to be recurrently utilized from now on:

\begin{dfn}
 Given a scenario $\Gamma=(X,\mc{C},O)$, let $\mathscr{H}$ be the 
extended hypergraph associated with it. We call the 2-section of 
$\mathscr{H}$ the \emph{extended compatibility graph} associated with 
$\Gamma$, and denote it by $\mathscr{G}$.
 \label{def:ext_comp_graph}
\end{dfn}

Now, as another corollary of Thm. \ref{thm:extended}, we 
have:

\begin{cor}
If a behavior $\mathrm{B}$ is noncontextual in the extended sense, then 
  there is an extended behavior  $\mathscr{B}$
for  $\mathrm{B} $ such that $\mathrm{P}_{\mathscr{B}}$ 
  belongs to the cut polytope of $\nabla \mathscr{G}$ , where $\mathscr{G}$ 
is the  extended compatibility 
graph of the scenario.
\label{cor:ext_comp_graph}
\end{cor}

%\st{We call the graph  $\mathscr{G}$ the 
%\emph{extended compatibility graph} of the scenario.}
%\begin{cor}
 %If each context in $\mathrm{H}$
 %has at most three measurements, a behavior $\mathrm{B}_{\mathrm{H}}$ is non-contextual in the extended sense if, and only if, 
 %the vector $\mathrm{P}_{\mathscr{B}}$ corresponding to the unique extended behavior $\mathscr{B}_{\mathscr{H}}$ for $\mathrm{B}_{\mathrm{H}}$ 
 %belongs to $\nabla \mathscr{G}$.
%\end{cor}

\subsection{Triangular elimination}

From valid inequalities for $\mathrm{CUT}\left(\nabla \mathrm{G}\right)$ 
it is possible to derive valid inequalities for $\mathrm{CUT}\left(\nabla 
\mathscr{G}\right)$ using the operation of \emph{triangular elimination}.

\begin{dfn}[Triangular Elimination for Graphs]
Let $\mathsf{G}=\left(\mathsf{V},\mathsf{E}\right)$ be a graph, $t$ an 
integer, and let $\mathsf{F}=\left\{u_iv_i \left| i =1, \ldots , 
t\right.\right\}$ be any subset of $\mathsf{E}$.
The graph $\mathsf{G}'=\left(\mathsf{V}',\mathsf{E}'\right)$ is a 
\emph{triangular elimination} of $\mathsf{G}$ with respect to $\mathsf{F}$ 
if $\mathsf{V}'=\mathsf{V} \cup \left\{w_1, w_2, \ldots, w_t\right\}$,
where $w_1, w_2, \ldots, w_t$ are new vertices not in $\mathsf{V}$, and 
$\mathsf{E}' \supseteq \left\{w_iu_i, w_iv_i \left| i = 1, \ldots ,  
t\right.\right\}$ and $\mathsf{E}'\cap \mathsf{E}= \mathsf{E}\setminus 
\mathsf{F}.$
 \label{def:triangular_elim_for_graphs}
\end{dfn} 

The graph $\mathsf{G}'$ is obtained from $\mathsf{G}$ by removing each 
edge $u_iv_i$ in $\mathsf{F}$ from $\mathsf{E}$ and replacing it with a 
new vertex $w_i$, which is connected to $u_i$ and $v_i$. Other edges 
connecting $w_i$ with other vertices other then $u_i$ and $v_i$ may or may 
not be added.
A simple example is shown in Fig. \ref{fig:te}.

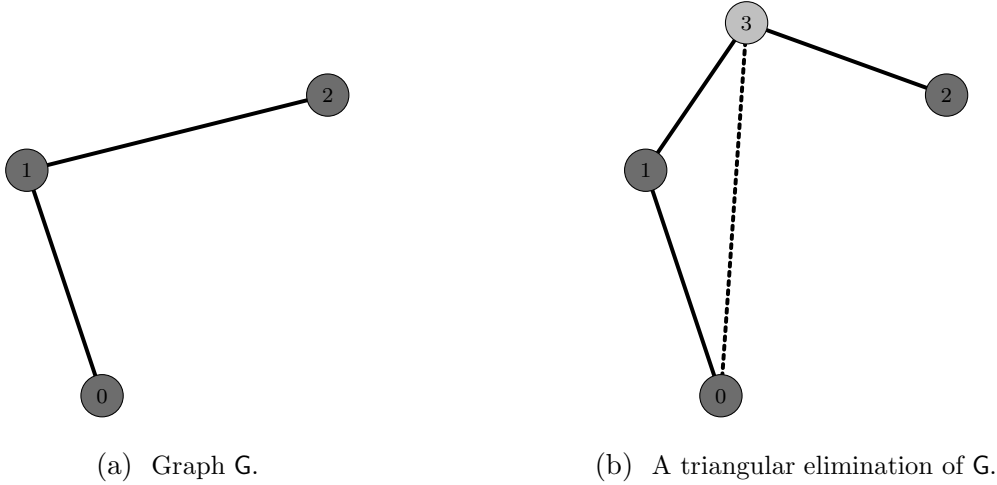
\begin{figure}[h!]
\begin{subfigure}[b]{.5\columnwidth}
  \centering
  \definecolor{ffwwzz}{rgb}{0.7529411764705882,0.7529411764705882,0.7529411764705882}
   \definecolor{wwzzff}{rgb}{0.4264705882352941,0.4264705882352941,0.4264705882352941}
  \begin{tikzpicture}[line cap=round,line join=round,>=triangle 45,x=1cm,y=1cm]
  \draw [line width=1.5pt] (4,-2)-- (3,1);
  \draw [line width=1.5pt] (3,1)-- (7,2);
  \begin{scriptsize}
  \draw [fill=wwzzff] (4,-2) circle (8pt);
  \draw[color=black] (4,-2) node {$0$};
  \draw [fill=wwzzff] (3,1) circle (8pt);
  \draw[color=black] (3,1) node {$1$};
  \draw [fill=wwzzff] (7,2) circle (8pt);
  \draw[color=black] (7,2) node {$2$};
  \end{scriptsize}
  \end{tikzpicture}
  \caption{\footnotesize{ Graph $\mathsf{G}$.}}
\end{subfigure}%
\begin{subfigure}[b]{.5\columnwidth}
  \centering
  \definecolor{ffwwzz}{rgb}{0.7529411764705882,0.7529411764705882,0.7529411764705882}
  \definecolor{wwzzff}{rgb}{0.4264705882352941,0.4264705882352941,0.4264705882352941}
  \begin{tikzpicture}[line cap=round,line join=round,>=triangle 45,x=1cm,y=1cm]
  \draw [line width=1.5pt] (15,-2)-- (14,1);
  \draw [line width=1.5pt] (14,1)-- (15.341264066064337,2.961439325753775);
  \draw [line width=1.5pt] (15.341264066064337,2.961439325753775)-- (18,2);
  \draw [line width=1.5pt, dotted] (15.341264066064337,2.961439325753775)-- (15,-2);
  \begin{scriptsize}
  \draw [fill=wwzzff] (15,-2) circle (8pt);
  \draw[color=black] (15,-2) node {$0$};
  \draw [fill=wwzzff] (14,1) circle (8pt);
  \draw[color=black] (14,1) node {$1$};
  \draw [fill=wwzzff] (18,2) circle (8pt);
  \draw[color=black] (18,2) node {$2$};
  \draw [fill=ffwwzz] (15.341264066064337,2.961439325753775) circle (8pt);
  \draw[color=black] (15.341264066064337,2.961439325753775) node {$3$};
  \end{scriptsize}
  \end{tikzpicture}
  \caption{\footnotesize{ A triangular elimination of $\mathsf{G}$.}}
\end{subfigure}
\caption{\footnotesize{A triangular elimination of graph $\mathsf{G}$ with 
respect to the edge $\left\{1,2\right\}$. This edge is removed, and a new 
vertex, labelled vertex $3$, is added.This new vertex is connected to both 
$1$ and $2$. The dashed edge connecting $0$ and $3$ may or may not be 
added. In both cases one ends with a valid triangular elimination 
for the graph $\mathsf{G}$}}
\label{fig:te}
\end{figure}

\begin{dfn}[Triangular elimination for inequalities]
Let $\mathsf{G}'=\left(\mathsf{V}', \mathsf{E}'\right)$ be a triangular 
elimination of $\mathsf{G}=\left(\mathsf{V}, \mathsf{E}\right)$ with 
respect to $\mathsf{F}=\left\{u_iv_i \left| i =1, \ldots , 
t\right.\right\}$, and suppose $A \in \mathbb{R}^{\mathsf{E}}$, $A' \in 
\mathbb{R}^{\mathsf{E}'}$, $b, b' \in \mathbb{R}$. The inequality $A' \cdot 
P' \leq b'$ is a \emph{triangular elimination} of inequality  $A \cdot P 
\leq b$ if it can be obtained from this last inequality  by summing  
positive multiples of inequalities 
\be
 -P_{u_iw_i} - P_{v_iw_i}- P_{u_iv_i} \leq 1,
 \label{eq:inq_te_1}
\ee
\be
 P_{u_iw_i} + P_{v_iw_i}- P_{u_iv_i} \leq 1
 \label{eq:inq_te_2}
\ee
or the other two inequalities obtained from  \eqref{eq:inq_te_2} by 
permuting $u_i, v_i$ and $w_i$.

\end{dfn}

\begin{prop}
\label{prop:tri}
 Let $\mathsf{G}'=\left(\mathsf{V}', \mathsf{E}'\right)$ be a triangular 
elimination of $\mathsf{G}=\left(\mathsf{V}, \mathsf{E}\right)$. 
 Let   $A \cdot P \leq b$ be a valid inequality for 
$\mathrm{CUT}\left(\mathsf{G}\right)$ and $A' \cdot P' \leq b'$ be a 
triangular elimination of
 $A \cdot P \leq b$. Then $A' \cdot P' \leq b'$ is valid for 
$\mathrm{CUT}\left(\mathsf{G}'\right)$.
\end{prop}

\textbf{Remark:} We should remark that for our own purposes the 
content of Prop. \ref{prop:tri} above is enough (see 
Corollary \ref{cor:necessary_condt}). Nonetheless, in 
Ref.~\cite{AIT06} the authors have shown that the other implication in 
Prop. \ref{prop:tri} is also true. It means that if $A^{\prime} \cdot 
P^{\prime} \leq b^{\prime}$ is a valid inequality for 
$\mathrm{CUT}\left(\mathsf{G}^{\prime}\right)$, then $A \cdot P \leq b$ is 
valid for $\mathrm{CUT}\left(\mathsf{G}\right)$, provided that 
$\mathsf{G}^{\prime}$ and $A^{\prime} \cdot P^{\prime} \leq b^{\prime}$ be 
triangular 
eliminations of $\mathsf{G}$ and $A \cdot P \leq b$ respectively.

%\subsection{Contexts with at most two measurements}

%When the contexts have at most two measurements, $\mathrm{H}$ is a simple graph and $\mathrm{H}=\mathrm{G}$. In this case, condition given by Theorem
%\ref{thm:nec_cut} is also necessary.

%\begin{thm}
% If $\mathrm{H}=\mathrm{G}$, a  behavior $\mathrm{B}_{\mathrm{H}}$ is noncontextual if, and only if,
% the vector $\mathrm{P}_{\mathrm{B}}$   belongs to the cut polytope of $\nabla \mathrm{G}$.
% \label{thm:nec_suf_cut}
%\end{thm}

%In this situation, we can find necessary conditions for extended contextuality using triangular elimination.

\subsection{Triangular elimination and extexted contextuality}

\begin{thm}
 Let $\Gamma$ be a 
compatibility scenario. If the compatibility hypergraph of $\Gamma$ 
coincides with its 2-section, \emph{i.e.} if $\mathrm{H} = \mathrm{G}$, 
then the extended compatibility graph $\mathscr{G}$ is a triangular 
elimination of the compatibility graph $\mathrm{G}$.
 Moreover, $\nabla \mathscr{G}$ is a triangular elimination of  $\nabla \mathrm{G}$.
 \label{thm:extd_graph_elimintion}
\end{thm}

\begin{proof}
 We start with $\mathrm{G}=(X,E(\mathrm{G}))$ and $x_1 \in 
\mathrm{X}$. Let $E_{x_1}=\left\{ x_1y_1, x_1y_2, \ldots , x_1y_n 
\right\} \subset E(\mathrm{G})$ be the 
 set of all edges incident to $x_1$ and let 
$\mathrm{G}_1 = (V(\mathrm{G}_1),E(\mathrm{G}_1))$ be the graph 
obtained from $\mathrm{G}$ in the following way:
 remove from $E(\mathrm{G})$ all edges in $E_{x_1}$,  
from $X$ remove the vertex $x_1$, add vertices $x_1^1, x_1^2, \ldots 
, x_1^n$, edges 
 $\left\{ x_1^1y_1, x_1^2y_2, \ldots , x_1^ny_n \right\}$ and all edges $x_1^kx_1^l$ with $1 \leq k < l \leq n$. 
 The graph $\mathrm{G}_1$ is a triangular elimination of $\mathrm{G}$. Take 
now $x_2 \in V(\mathrm{G}_1) \setminus \left\{x_1^1, x_1^2, \ldots , 
x_1^n\right\}$
 and repeat the same procedure, obtaining graph $\mathrm{G}_2$. Proceeding 
analogously for every vertex in $\mathrm{X}$, and since it is 
finite, we get 
 $\mathscr{G}$ in the last step. Similar argument can be used with $\nabla \mathrm{G}$.
\end{proof}

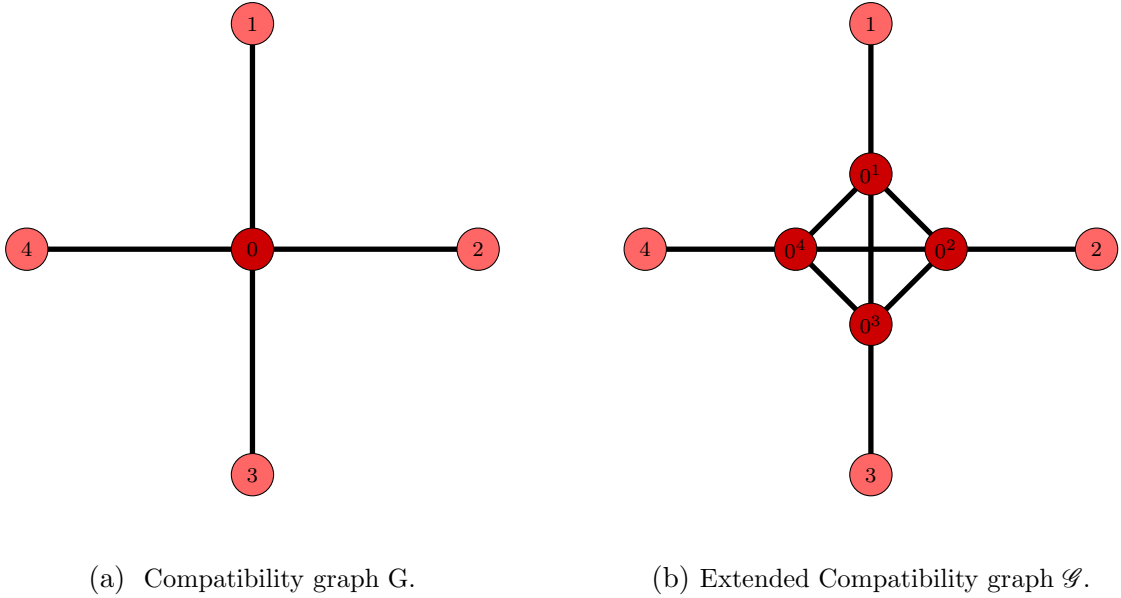
\begin{figure}[h!]
\begin{subfigure}{.5\columnwidth}
  \begin{center}
\definecolor{ffefdv}{rgb}{1,0.4,0.4}
\definecolor{ccqqqq}{rgb}{0.8,0,0}
\begin{tikzpicture}[line cap=round,line join=round,>=triangle 45,x=1cm,y=1cm]
\draw [line width=2pt] (3,-1)-- (3,2);
\draw [line width=2pt] (3,-1)-- (6,-1);
\draw [line width=2pt] (3,-1)-- (3,-4);
\draw [line width=2pt] (0,-1)-- (3,-1);
\begin{scriptsize}
\draw [fill=ccqqqq] (3,-1) circle (8pt);
\draw[color=black] (3,-1) node {$0$};
\draw [fill=ffefdv] (3,2) circle (8pt);
\draw[color=black] (3,2) node {$1$};
\draw [fill=ffefdv] (6,-1) circle (8pt);
\draw[color=black] (6,-1) node {$2$};
\draw [fill=ffefdv] (0,-1) circle (8pt);
\draw[color=black] (0,-1) node {$4$};
\draw [fill=ffefdv] (3,-4) circle (8pt);
\draw[color=black] (3,-4) node {$3$};
\end{scriptsize}
\end{tikzpicture}
  \end{center}
  \caption{\footnotesize{ Compatibility graph $\mathrm{G}$.}}
\end{subfigure}%
\begin{subfigure}{.5\columnwidth}
  \begin{center}
\definecolor{ffefdv}{rgb}{1,0.4,0.4}
\definecolor{ccqqqq}{rgb}{0.8,0,0}
\begin{tikzpicture}[line cap=round,line join=round,>=triangle 45,x=1cm,y=1cm]
\draw [line width=2pt] (17,0)-- (16,-1);
\draw [line width=2pt] (16,-1)-- (17,-2);
\draw [line width=2pt] (17,-2)-- (18,-1);
\draw [line width=2pt] (18,-1)-- (17,0);
\draw [line width=2pt] (16,-1)-- (18,-1);
\draw [line width=2pt] (14,-1)-- (16,-1);
\draw [line width=2pt] (17,0)-- (17,2);
\draw [line width=2pt] (18,-1)-- (20,-1);
\draw [line width=2pt] (17,-4)-- (17,-2);
\draw [line width=2pt] (17,0)-- (17,-2);
\begin{scriptsize}
\draw [fill=ffefdv] (14,-1) circle (8pt);
\draw[color=black] (14,-1) node {$4$};
\draw [fill=ffefdv] (20,-1) circle (8pt);
\draw[color=black] (20,-1) node {$2$};
\draw [fill=ffefdv] (17,2) circle (8pt);
\draw[color=black] (17,2) node {$1$};
\draw [fill=ffefdv] (17,-4) circle (8pt);
\draw[color=black] (17,-4) node {$3$};
\draw [fill=ccqqqq] (16,-1) circle (8pt);
\draw[color=black] (16,-1) node {$0^4$};
\draw [fill=ccqqqq] (18,-1) circle (8pt);
\draw[color=black] (18,-1) node {$0^2$};
\draw [fill=ccqqqq] (17,-2) circle (8pt);
\draw[color=black] (17,-2) node {$0^3$};
\draw [fill=ccqqqq] (17,0) circle (8pt);
\draw[color=black] (17,0) node {$0^1$};
\end{scriptsize}
\end{tikzpicture}
  \end{center}
  \caption{\footnotesize{Extended Compatibility graph $\mathscr{G}$.}}
\end{subfigure}
\caption{\footnotesize{The extended compatibility graph $\mathscr{G}$ is a 
triangular elimination of $\mathrm{G}$.
In this case, after relabelling vertex $0$ as $0_1$, we remove edges 
$\left\{0,2\right\}$, $\left\{0,3\right\}$ and $\left\{0,4\right\}$. We add 
vertex
$0_2$, connected to $0_1$ and $2$, vertex $0_3$, connected to $0_1$ and $3$ and vertex $4$, connected to $0_1$ and $0_4$. We also add all edges between
the vertices $0_i$ that are missing.}}
\label{fig:fig}
\end{figure}

As a direct consequence of Prop. \ref{prop:tri} and Thm. \ref{thm:extd_graph_elimintion}, we have:

\begin{cor}
Given a compatibility scenario $\Gamma$, suppose that 
$\mathrm{H}=\mathrm{G}$. A necessary condition for the behavior   
$\mathrm{B}$ to be noncontextual in the extended sense is that for every
extended behavior $\mathscr{B}$ and for every inequality
valid for $\mathrm{CUT}\left(\nabla \mathrm{G}\right)$,  its triangular eliminations are satisfied by 
 the vector $\mathrm{P}_{\mathscr{B}}$ corresponding to   $\mathscr{B}$.
\label{cor:necessary_condt} 
\end{cor}

It is important to notice that  terms of the form of inequality  \eqref{eq:inq_te_2} added to $A \cdot P \leq b$ will be satisfied at equality if the behaviors are perfectly non-disturbing. Hence, there is one triangular elimination of 
$A \cdot P \leq b$ that is tight and reduces to the original inequality for non-disturbing behaviors.

When $\mathscr{B}$ is unique, we obtain a simple necessary condition for noncontextuality in the extended sense.
We calculate $\mathrm{P}_{\mathscr{B}}$ and substitute its entries in the inequalities for 
$\mathrm{CUT}\left(\nabla \mathscr{G}\right)$ obtained from the inequalities for
$\mathrm{CUT}\left(\nabla \mathrm{G}\right)$ via triangular elimination. If we find that some of them are not satisfied, we can conclude that
$\mathrm{B}$ is contextual in the extended sense.

In the case $\mathscr{B}$ is not unique, it may be impractical to determine all possible 
$\mathrm{P}_{\mathscr{B}}$ so we can not test directly if these vectors satisfy all triangular eliminations of a given inequality for
$\mathrm{CUT}\left(\nabla \mathrm{G}\right)$ or not.  Nevertheless, Thm. \ref{thm:necessary_max_coup} will help us circumvent this difficulty.

If  $A' \cdot P' \leq b'$ is a triangular elimination of $A \cdot P \leq b$, then  the left-hand-side can be written as a sum of
two terms $A'\cdot P' =  A_1 \cdot P_1 + A_2\cdot P_2$,  where $P_1$ is the projection of $P'$ that contains the entries depending only on the contexts in 
$\mathscr{H}$ that come from the contexts in $\mathrm{H}$ and $P_2$ is the projection of $P'$ that contains the terms
depending only on the contexts consisting on random variables that represent the same measurement. From $\mathrm{P}_{\mathrm{B}}$
we calculate $P_1$. To calculate $P_2$ explicitly we have to determine the maximal couplings for each pair of variables
 that represent the same measurement, which can be a hard task. Instead of doing this, we use the necessary condition 
satisfied for all maximal couplings presented in theorem \ref{thm:necessary_max_coup} to calculate which value of 
$A_2\cdot P_2$ is the worst, respecting the condition of maximal couplings. This proves the following:

\begin{thm}
 Let $A'\cdot P' =  A_1 \cdot P_1 + A_2\cdot P_2 \leq b'$ be a valid inequality for $\mathrm{CUT}\left(\nabla \mathscr{G}\right)$. Let $m$ be the minimum  of $A_2\cdot P_2$ over all 
 possible values of $P_2$ satisfying conditions given in Thm. \ref{thm:necessary_max_coup}. If 
 \be A_1 \cdot P_{\mathrm{B}} + m > b'\ee
 $\mathrm{P}_{\mathrm{B}}$ is contextual in the extended sense.
\end{thm}

This gives a necessary condition for extended contextuality  that can be applied in any compatibility scenario.

\section{The $I_{3322}$ inequality}
\label{sec:I3322}

Our first example is the $(3,3,2,2)$ Bell scenario \cite{Foissart81,CG04}, 
where two distinct  parties perform three measurements each, each 
measurement with two outcomes. In this case each context has 
exactly two measurements, one form each party. With our notation, 
it means that
this scenario is described by \be 
\Gamma= \lbrace \{ A_1,A_2,A_3,B_1,B_2,B_3 \}, \{A_i B_j\}_{i \neq j}, 
\{-1,1\} \rbrace \ee and $\mathrm{H} = \mathrm{G}$.
The compatibility graph of this scenario is the complete bipartite graph $K_{3,3}$, shown in Fig. \ref{fig:I3322}.

\begin{figure}[h!]
  \begin{center}
 \definecolor{ffqqff}{rgb}{1,0,1}\definecolor{ffdxqq}{rgb}{1,0.8431372549019608,0}
 \definecolor{ffqqqq}{rgb}{1,0,0}\definecolor{qqzzqq}{rgb}{0,0.6,0}
 \definecolor{ubqqys}{rgb}{0.49411764705882354,0.2,0.7098039215686274}
 \definecolor{qqqqff}{rgb}{0.25,0.25,1}
 \begin{tikzpicture}[line cap=round,line join=round,>=triangle 45,x=1cm,y=1cm]
 \draw [line width=2pt] (2,2)-- (2,-1);
 \draw [line width=2pt] (5,2)-- (5,-1);
 \draw [line width=2pt] (8,2)-- (8,-1);
 \draw [line width=2pt] (2,2)-- (5,-1);
 \draw [line width=2pt] (2,2)-- (8,-1);
 \draw [line width=2pt] (5,2)-- (8,-1);
 \draw [line width=2pt] (8,2)-- (5,-1);
 \draw [line width=2pt] (5,2)-- (2,-1);
 \draw [line width=2pt] (8,2)-- (2,-1);
 \draw [fill=qqqqff] (2,2) circle (10pt);
 \draw[color=black] (2,2) node {$A_1$};
 \draw [fill=ubqqys] (5,2) circle (10pt);
 \draw[color=black] (5,2) node {$A_2$};
 \draw [fill=qqzzqq] (8,2) circle (10pt);
 \draw[color=black] (8,2) node {$A_3$};
 \draw [fill=ffqqqq] (2,-1) circle (10pt);
 \draw[color=black] (2,-1) node {$B_1$};
 \draw [fill=ffdxqq] (5,-1) circle (10pt);
 \draw[color=black] (5,-1) node {$B_2$};
 \draw [fill=ffqqff] (8,-1) circle (10pt);
 \draw[color=black] (8,-1) node {$B_3$};
\end{tikzpicture}
  \end{center}
  \caption{\footnotesize{Compatibility graph $\mathrm{G}$ of the 
$(3,3,2,2)$ Bell scenario. Measurements of first party are labeled $A_1, 
A_2, A_3$ and measurements of 
  the second party are labelled $B_1, B_2, B_3$.}}
  \label{fig:I3322}
\end{figure}
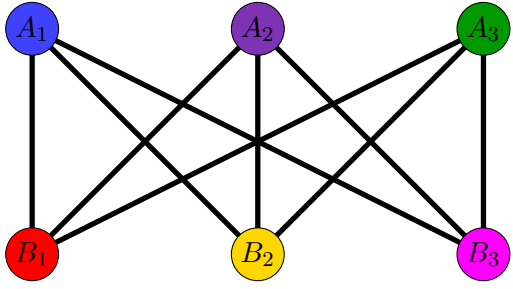

One of the facets of $\mathrm{CUT}\left(\nabla \mathrm{G}\right)$ is given by the so called $\mathrm{I}_{3322}$ inequality \cite{Foissart81,CG04}:
\begin{multline}
\label{eq:I3322}
\left\langle A_1 \right\rangle + \left\langle A_2 \right\rangle + \left\langle B_1 \right\rangle +\left\langle B_2 \right\rangle
-\left\langle A_1 B_1 \right\rangle-\left\langle A_1B_2 \right\rangle  \\
-\left\langle A_1 B_3\right\rangle
-\left\langle A_2B_1 \right\rangle-\left\langle A_2B_2 \right\rangle  + \left\langle A_2B_3 \right\rangle
-\left\langle A_3B_1 \right\rangle +\left\langle A_3B_2 \right\rangle \leq 4
\end{multline}

The extended compatibility graph $\mathscr{G}$ of this scenario is shown in Fig. \ref{fig:I3322_ext}. Each vertex $A_i$ becomes three
new vertices $A_i^1, A_i^2, A_i^3$ in $\mathscr{G}$, and similar for each $B_i$. Vertices $A_i^j$ and $B_j^i$ are connected. 
The vertices $A_i^1, A_i^2, A_i^3$ are connected for each $i$ and similar for  $B_i^1, B_i^2, B_i^3$.

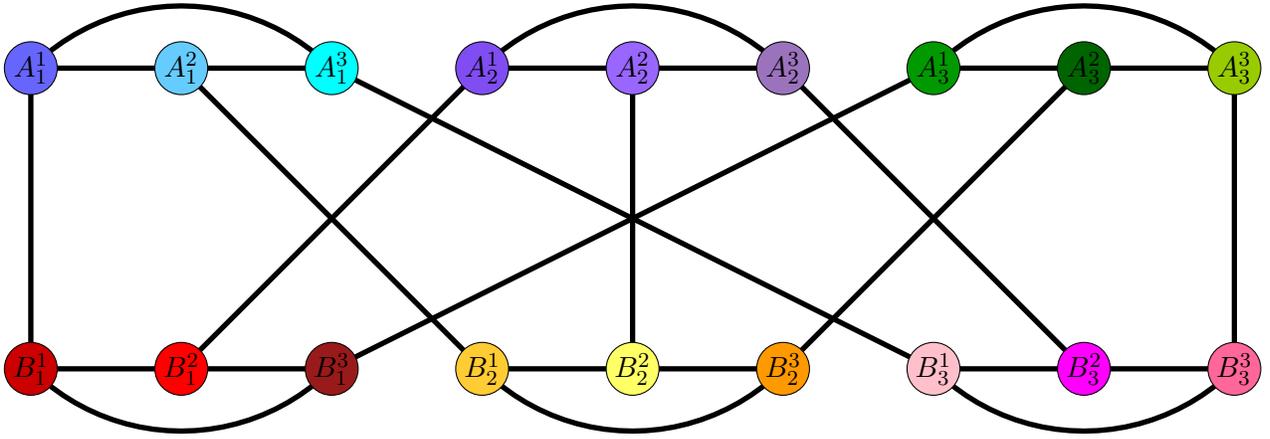
\begin{figure}[h!]
  \begin{center}
\definecolor{ffwwzz}{rgb}{1,0.4,0.6}
\definecolor{ffqqff}{rgb}{1,0,1}
\definecolor{ffcqcb}{rgb}{1,0.7529411764705882,0.796078431372549}
\definecolor{ffzzqq}{rgb}{1,0.6,0}
\definecolor{ffffww}{rgb}{1,1,0.4}
\definecolor{ffcctt}{rgb}{1,0.8,0.2}
\definecolor{yqqqqq}{rgb}{0.6019607843137255,0.1,0.1}
\definecolor{ffqqqq}{rgb}{1,0,0}
\definecolor{ccqqqq}{rgb}{0.8,0,0}
\definecolor{zzccqq}{rgb}{0.6,0.8,0}
\definecolor{qqwuqq}{rgb}{0,0.39215686274509803,0}
\definecolor{qqzzqq}{rgb}{0,0.6,0}
\definecolor{wyuqya}{rgb}{0.60784313725490196,0.45098039215686274,0.7411764705882353}
\definecolor{zzwwff}{rgb}{0.6,0.4,1}
\definecolor{wwqqcc}{rgb}{0.5,0.3,0.95}
\definecolor{qqffff}{rgb}{0,1,1}
\definecolor{wwccff}{rgb}{0.4,0.8,1}
\definecolor{ududff}{rgb}{0.40196078431372547,0.40196078431372547,1}
\begin{tikzpicture}[line cap=round,line join=round,>=triangle 45,x=1cm,y=1cm]
\draw [line width=2pt] (1,2)-- (1,-2);
\draw [line width=2pt] (1,2)-- (3,2);
\draw [line width=2pt] (3,2)-- (5,2);
\draw [line width=2pt] (3,2)-- (7,-2);
\draw [line width=2pt] (9,2)-- (9,-2);
\draw [line width=2pt] (11,2)-- (15,-2);
\draw [line width=2pt] (7,2)-- (3,-2);
\draw [line width=2pt] (5,2)-- (13,-2);
\draw [line width=2pt] (13,2)-- (5,-2);
\draw [line width=2pt] (15,2)-- (11,-2);
\draw [line width=2pt] (17,2)-- (17,-2);
\draw [line width=2pt] (7,2)-- (9,2);
\draw [line width=2pt] (9,2)-- (11,2);
\draw [line width=2pt] (13,2)-- (15,2);
\draw [line width=2pt] (15,2)-- (17,2);
\draw [line width=2pt] (13,-2)-- (15,-2);
\draw [line width=2pt] (15,-2)-- (17,-2);
\draw [line width=2pt] (7,-2)-- (9,-2);
\draw [line width=2pt] (9,-2)-- (11,-2);
\draw [line width=2pt] (1,-2)-- (3,-2);
\draw [line width=2pt] (3,-2)-- (5,-2);
\draw [shift={(3,0)},line width=2pt]  plot[domain=0.7853981633974483:2.356194490192345,variable=\t]({1*2.8284271247461903*cos(\t r)+0*2.8284271247461903*sin(\t r)},{0*2.8284271247461903*cos(\t r)+1*2.8284271247461903*sin(\t r)});
\draw [shift={(9,0)},line width=2pt]  plot[domain=0.7853981633974483:2.356194490192345,variable=\t]({1*2.8284271247461903*cos(\t r)+0*2.8284271247461903*sin(\t r)},{0*2.8284271247461903*cos(\t r)+1*2.8284271247461903*sin(\t r)});
\draw [shift={(15,0)},line width=2pt]  plot[domain=0.7853981633974483:2.356194490192345,variable=\t]({1*2.8284271247461903*cos(\t r)+0*2.8284271247461903*sin(\t r)},{0*2.8284271247461903*cos(\t r)+1*2.8284271247461903*sin(\t r)});
\draw [shift={(3,0)},line width=2pt]  plot[domain=3.9269908169872414:5.497787143782138,variable=\t]({1*2.8284271247461903*cos(\t r)+0*2.8284271247461903*sin(\t r)},{0*2.8284271247461903*cos(\t r)+1*2.8284271247461903*sin(\t r)});
\draw [shift={(9,0)},line width=2pt]  plot[domain=3.9269908169872414:5.497787143782138,variable=\t]({1*2.8284271247461903*cos(\t r)+0*2.8284271247461903*sin(\t r)},{0*2.8284271247461903*cos(\t r)+1*2.8284271247461903*sin(\t r)});
\draw [shift={(15,0)},line width=2pt]  plot[domain=3.9269908169872414:5.497787143782138,variable=\t]({1*2.8284271247461903*cos(\t r)+0*2.8284271247461903*sin(\t r)},{0*2.8284271247461903*cos(\t r)+1*2.8284271247461903*sin(\t r)});
\draw [fill=ududff] (1,2) circle (10pt);
\draw[color=black] (1,2) node {$A_1^1$};
\draw [fill=wwccff] (3,2) circle (10pt);
\draw[color=black] (3,2) node {$A_1^2$};
\draw [fill=qqffff] (5,2) circle (10pt);
\draw[color=black] (5,2) node {$A_1^3$};
\draw [fill=wwqqcc] (7,2) circle (10pt);
\draw[color=black] (7,2) node {$A_2^1$};
\draw [fill=zzwwff] (9,2) circle (10pt);
\draw[color=black] (9,2) node {$A_2^2$};
\draw [fill=wyuqya] (11,2) circle (10pt);
\draw[color=black] (11,2) node {$A_2^3$};
\draw [fill=qqzzqq] (13,2) circle (10pt);
\draw[color=black] (13,2) node {$A_3^1$};
\draw [fill=qqwuqq] (15,2) circle (10pt);
\draw[color=black] (15,2) node {$A_3^2$};
\draw [fill=zzccqq] (17,2) circle (10pt);
\draw[color=black] (17,2) node {$A_3^3$};
\draw [fill=ccqqqq] (1,-2) circle (10pt);
\draw[color=black] (1,-2) node {$B_1^1$};
\draw [fill=ffqqqq] (3,-2) circle (10pt);
\draw[color=black] (3,-2) node {$B_1^2$};
\draw [fill=yqqqqq] (5,-2) circle (10pt);
\draw[color=black] (5,-2) node {$B_1^3$};
\draw [fill=ffcctt] (7,-2) circle (10pt);
\draw[color=black] (7,-2) node {$B_2^1$};
\draw [fill=ffffww] (9,-2) circle (10pt);
\draw[color=black] (9,-2) node {$B_2^2$};
\draw [fill=ffzzqq] (11,-2) circle (10pt);
\draw[color=black] (11,-2) node {$B_2^3$};
\draw [fill=ffcqcb] (13,-2) circle (10pt);
\draw[color=black] (13,-2) node {$B_3^1$};
\draw [fill=ffqqff] (15,-2) circle (10pt);
\draw[color=black] (15,-2) node {$B_3^2$};
\draw [fill=ffwwzz] (17,-2) circle (10pt);
\draw[color=black] (17,-2) node {$B_3^3$};
\end{tikzpicture}
  \end{center}
  \caption{\footnotesize{Extended compatibility graph $\mathscr{G}$.}}
  \label{fig:I3322_ext}
\end{figure}

Applying  triangular elimination in the $\mathrm{I}_{3322}$ inequality, we can derive the following valid inequality for 
$\mathrm{CUT}\left(\nabla \mathscr{G}\right)$

\begin{multline}
 \left\langle A^1_1 \right\rangle + \left\langle A^1_2 \right\rangle + \left\langle B^1_1 \right\rangle +\left\langle B^1_2 \right\rangle
-\left\langle A^1_1 B^1_1 \right\rangle-\left\langle A^2_1B^1_2 \right\rangle-\left\langle A^3_1 B^1_3\right\rangle 
-\left\langle A^1_2B^2_1 \right\rangle \\ -\left\langle A^2_2B^2_2 \right\rangle+\left\langle A^3_2B^2_3 \right\rangle
 -\left\langle A^1_3B^3_1 \right\rangle +\left\langle A^2_3B^3_2 \right\rangle  
+ \left\langle A^1_1A_1^2 \right\rangle + \left\langle A^1_1A_1^3 \right\rangle + \left\langle A^1_2 A_2^2 \right\rangle  \\ + 
\left\langle A^1_2 A_2^3 \right\rangle
+ \left\langle A^1_3A_3^2 \right\rangle
+ \left\langle B^1_1B_1^2 \right\rangle + \left\langle B^1_1B_1^3 \right\rangle + \left\langle B^1_2B_2^2 \right\rangle + \left\langle B^1_2 B^1_3 \right\rangle
+ \left\langle B^1_3B_3^2 \right\rangle  \leq 14
\label{eq:I3322_ext}
\end{multline}

This inequality is tight and reduces to Ineq. \eqref{eq:I3322} for non-disturbing behaviors, since in this particular case we have that 
$\left\langle A^j_iA_i^k \right\rangle=1$ and $\left\langle B^j_iB_i^k \right\rangle=1$ for every $i, j,k$.

In this scenario, each measurement has two outcomes and belongs to 
three contexts, therefore each behavior $\mathrm{B}$ has a unique 
extended behavior $\mathscr{B}$ corresponding to it.  This, in turn, implies the following result:

\begin{cor}
 A necessary condition for extended noncontextuality of a behavior  $\mathrm{B}$  in the    $(3,3,2,2)$ Bell scenario is that the unique extended behavior
of  $\mathrm{B}$  satisfies the triangular elimination of the $\mathrm{I}_{3322}$ inequality given by Eq.
 \eqref{eq:I3322_ext}.
           
\end{cor}

%\subsection{Experimental Demonstration of Extended Contextuality for the $I_{3322}$ Inequality?}

\section{Chained Inequalitites}
\label{sec:chain_ineq}

We consider now the $(n,n,2,2)$ Bell scenario with $2$ parties, $n$ 
measurements per party, each measurements with $2$ outcomes. Also in this 
case each context has exactly two measurements, one from each party, and 
$\mathrm{H}= \mathrm{G}$.  Once again, sticking to our notation, we 
describe such a scenario with \be \Gamma=\lbrace 
\{A_1,...,A_n,B_1,...,B_n\}, \{A_i A_j\}_{i \neq j}, \{-1,1\}    \rbrace 
\ee The compatibility graph $\mathrm{G}$ is the complete bipartite graph 
$K_{n,n}$. A family of noncontextuality inequalities for these 
scenarios consists of the so called \emph{Chained Inequalities} 
\cite{BC90}, given by

\be
\left\langle A_1B_2\right\rangle + \left\langle B_1A_2\right\rangle + \ldots + \left\langle B_{n-1}A_n\right\rangle 
+\left\langle  A_nB_n \right\rangle - \left\langle B_n A_1\right\rangle \leq 2n-2.
\label{eq:chained}\ee

Each vertex $A_i$ becomes $n$
new vertices $A_i^1, A_i^2, \ldots , A_i^n$ in  the extended compatibility graph $\mathscr{G}$, and similar for each $B_i$. Vertices $A_i^j$ and $B_j^i$ are connected. 
The vertices $A_i^1, A_i^2, \ldots , A_i^n$ are connected for each $i$ and similar for $B_i^1, B_i^2, \ldots , B_i^n$.

Applying  triangular elimination in the  inequality \eqref{eq:chained}, we can derive the following valid inequality for 
$\mathrm{CUT}\left(\nabla \mathscr{G}\right)$, which is tight and reduces to Ineq. \eqref{eq:chained} for no-disturbing behaviors:

\begin{multline}
\left\langle A_1^2B^1_2\right\rangle + \left\langle B^2_1A^1_2\right\rangle + \ldots + \left\langle B^n_{n-1}A^{n-1}_n\right\rangle 
+\left\langle  A^n_nB^n_n \right\rangle - \left\langle B^1_n A^n_1\right\rangle 
+\left\langle  A^1_1A^n_1 \right\rangle 
+\left\langle  A^1_2A^2_2 \right\rangle + \ldots  \\
+\left\langle  A^{n-1}_nA^n_n \right\rangle
+\left\langle  B^1_1B^2_1 \right\rangle
+\left\langle  B^2_2B^3_2 \right\rangle +\ldots
+\left\langle  B^1_nB^n_n \right\rangle
\leq  4n-2. 
\label{eq:chained_ext}
\end{multline}

In this scenario, each measurement belongs to $n$ contexts, therefore each 
behavior $\mathrm{B}$ may have several extended behaviors $\mathscr{B}$ 
corresponding to it. Given such $\mathscr{B}$, we construct the 
vector $\mathrm{P}_{\mathscr{B}}$. Let $P_1$ be the projection of 
$\mathrm{P}_{\mathscr{B}}$ over the entries corresponding to contexts
$A_i^jB_j^i$ and $P_2$ be the projection of $\mathrm{P}_{\mathscr{B}}$ 
over the entries corresponding to contexts $A_i^jA_i^k$ and $B_i^jB_i^k$. 
$P_1$ depends only in $\mathrm{P}_{\mathrm{B}}$ and hence is the same for 
all extended behaviors $\mathrm{P}_{\mathscr{B}}$. The projection $P_2$ 
depends  on the choice of  maximal coupling for each pair $A_i^jA_i^k$ and 
$B_i^jB_i^k$.

The left-hand side of inequality \eqref{eq:chained_ext} can be divided in two parts. The first part contains the terms
\begin{widetext}
\be \left\langle A_1^2B^1_2\right\rangle + \left\langle 
B^2_1A^1_2\right\rangle + \ldots + \left\langle 
B^n_{n-1}A^{n-1}_n\right\rangle 
+\left\langle  A^n_nB^n_n \right\rangle - \left\langle B^1_n 
A^n_1\right\rangle \ee
\end{widetext}
and depends only on $P_1$, and hence only on $\mathrm{P}_{\mathrm{B}}$. The second part 
contains the terms
\begin{widetext}
\be
\left\langle  A^1_1A^n_1 \right\rangle
+\left\langle  A^1_2A^2_2 \right\rangle + \ldots 
+\left\langle  A^{n-1}_nA^n_n \right\rangle
+\left\langle  B^1_1B^2_1 \right\rangle
+\left\langle  B^2_2B^3_2 \right\rangle +\ldots
+\left\langle  B^1_nB^n_n \right\rangle
\label{eq:chained_ext_part2}
\ee
\end{widetext}
and depends only on $P_2$. No matter which extended behavior we have, the projection $P_2$ must necessarily satisfy the constraint given in Thm.
\ref{thm:necessary_max_coup}. Let $m$ be the minimum of the second term \eqref{eq:chained_ext_part2} over all vectors $P_2$ satisfying 
Thm.
\ref{thm:necessary_max_coup}. 

\begin{cor}
 A necessary condition for extended noncontextuality  of a behavior $\mathrm{B}$ in the $(n,n,2,2)$ Bell scenario is that the inequality
 \be
  \left\langle A_1^2B^1_2\right\rangle + \left\langle B^2_1A^1_2\right\rangle + \ldots + \left\langle B^n_{n-1}A^{n-1}_n\right\rangle 
+\left\langle  A^n_nB^n_n \right\rangle 
 - \left\langle B^1_n A^n_1\right\rangle + m \leq  
2\left(2n-1\right)
 \ee
is satisfied by the projection $P_1$ of the extended behaviors $\mathscr{B}$ for $\mathrm{B}$.
\end{cor}

%\subsection{Experimental Demonstration of Extended Contextuality for the Chained Inequalities?}

\section{Scenarios with contexts with more than three measurements}
\label{sec:more_than_3}

When there are contexts with more then three 
measurements, $\mathrm{H} \neq \mathrm{G}$ and $\mathscr{G}$ is not 
a triangular elimination of 
$\mathrm{G}$. Nevertheless we can still generate valid inequalities 
for $\mathrm{CUT}\left(\nabla \mathscr{G}\right)$ from valid inequalities 
for $\mathrm{CUT}\left(\nabla \mathrm{G}\right)$ using two strategies: the 
first one is to use a graph operation called \emph{vertex 
splitting}~\cite{AIT06,DL97,BM86};
the second one is to use triangular elimination combined with a graph 
operation called \emph{edge contraction}~\cite{AIT06,DL97,BM86}.

\subsection{Vertex splitting}

\begin{dfn}[Vertex splitting for graphs]
Let $\mathsf{G}=\left(\mathsf{V}, \mathsf{E}\right)$ be a graph, $w \in 
\mathsf{V}$ and $\left(\mathsf{S}, \mathsf{T}, \mathsf{B}\right)$
be a partition of the neighbours of $w$. The graph $\mathsf{G}'=\left(\mathsf{V}', \mathsf{E}'\right)$ is obtained from $\mathsf{G}$ 
by \emph{splitting} vertex $w$ into $s$ and $t$, for $s,t \notin 
\mathsf{V}$,  with respect to the 
partition $\left(\mathsf{S}, \mathsf{T}, \mathsf{B}\right)$
if $$\mathsf{V}'= \left(\mathsf{V} \setminus \{w\} \right) \cup \left\{s,t\right\}$$ and
\be \mathsf{E}'= \left(\mathsf{E} \setminus \delta\left(w\right)\right) \cup \left(s: S \cup B\right) \cup \left(t: T \cup B\right) \cup \left\{st\right\},\ee
where $\delta\left(w\right)$ is the set of neighbours of $w$, $\left(s: S \cup B\right)$ is the set of all edges connecting $s$ to the vertices in $S \cup B$ and
$\left(t: T \cup B\right)$ is the set of all edges connecting $t$ to the vertices in $T \cup B$. 
\end{dfn}

In other words, the graph $\mathsf{G}'$ is the graph obtained from 
$\mathsf{G}$ removing the vertex $w$ and replacing it by vertices $s$ and 
$t$, which are connected.
The vertices in $S$ are connected only to $s$, the vertices in $T$ are connected only to $t$ and the vertices in $B$ are connected to both $s$ and $t$.
Figures \ref{fig:splitting1}-\ref{fig:splitting2} illustrate a 
simple example of this operation.

\begin{figure}[h!]
\begin{subfigure}[b]{.5\columnwidth}
  \begin{center}
\definecolor{ffwwqq}{rgb}{1,0.6,0}
\definecolor{ffdxqq}{rgb}{1,0.8431372549019608,0}
\definecolor{zzwwqq}{rgb}{0.8,0.6431372549019608,0}
\definecolor{ffefdv}{rgb}{1,0.9372549019607843,0.8352941176470589}
\begin{tikzpicture}[line cap=round,line join=round,>=triangle 45,x=1cm,y=1cm]
\draw [line width=2pt] (1,0)-- (5,2);
\draw [line width=2pt] (1,0)-- (5,0);
\draw [line width=2pt] (1,0)-- (5,-2);
\draw [line width=1pt, dashed] (5,2) circle (0.7cm);
\draw [line width=1pt, dashed] (5,0) circle (0.7cm);
\draw [line width=1pt, dashed] (5,-2) circle (0.7cm);
\draw [fill=ffefdv] (1,0) circle (10pt);
\draw[color=black] (1,0) node {$w$};
\draw [fill=zzwwqq] (5,2) circle (10pt);
\draw [fill=ffdxqq] (5,0) circle (10pt);
\draw [fill=ffwwqq] (5,-2) circle (10pt);
\draw[color=black] (6,2) node {$S$};
\draw[color=black] (6,0) node {$B$};
\draw[color=black] (6,-2) node {$T$};
\end{tikzpicture}
  \end{center}
  \caption{\footnotesize{Graph $\mathsf{G}$, vertex $w$ and partition 
$\left(\mathsf{S}, \mathsf{T}, \mathsf{B}\right)$ of 
$\delta{(w)}$.}\label{fig:splitting1}}
\end{subfigure}%
\begin{subfigure}[b]{.5\columnwidth}
  \begin{center}
\definecolor{ffwwqq}{rgb}{1,0.6,0}
\definecolor{ffdxqq}{rgb}{1,0.8431372549019608,0}
\definecolor{zzwwqq}{rgb}{0.8,0.6431372549019608,0}
\definecolor{ffefdv}{rgb}{1,0.9372549019607843,0.8352941176470589}
\begin{tikzpicture}[line cap=round,line join=round,>=triangle 45,x=1cm,y=1cm]
\draw [line width=2pt] (11,1)-- (15,2);
\draw [line width=2pt] (11,-1)-- (15,-2);
\draw [line width=2pt] (11,-1)-- (11,1);
\draw [line width=2pt] (11,1)-- (15,0);
\draw [line width=2pt] (11,-1)-- (15,0);
\draw [fill=ffefdv] (11,1) circle (10pt);
\draw[color=black] (11,1) node {$s$};
\draw [fill=ffefdv] (11,-1) circle (10pt);
\draw[color=black] (11,-1) node {$t$};
\draw [fill=zzwwqq] (15,2) circle (10pt);
\draw [fill=ffdxqq] (15,0) circle (10pt);
\draw [fill=ffwwqq] (15,-2) circle (10pt);
\end{tikzpicture}
  \end{center}
  \caption{\footnotesize{Vertex splitting of  $\mathsf{G}$ with respect 
to $w$ and the partition $\left(\mathsf{S}, \mathsf{T}, 
\mathsf{B}\right)$.} \label{fig:splitting2}}
\end{subfigure}
% \caption{The extended compatibility graph $\mathscr{G}$ is a triangular elimination of $\mathrm{G}$.
% In this case, after relabeling vertex $0$ as $0_1$, we remove edges $\left\{0,2\right\}$, $\left\{0,3\right\}$ and $\left\{0,4\right\}$. We add vertex
% $0_2$, connected to $0_1$ and $2$, vertex $0_3$, connected to $0_1$ and $3$ and vertex $4$, connected to $0_1$ and $0_4$. We also add all edges between
% the vertices $0_i$ that are missing.}
\end{figure}

\begin{dfn}[Vertex splitting for inequalities]
 Let $\mathsf{G}=\left(\mathsf{V}, \mathsf{E}\right)$ be a graph, $w \in \mathsf{V}$, $\left(\mathsf{S}, \mathsf{T}, \mathsf{B}\right)$
be a partition of the neighbours of $w$ and $A \cdot P \leq b$ be an inequality valid for $\mathrm{CUT}\left(\mathsf{G}\right)$.
Assume without loss of generality that $\sum_{v \in \mathsf{T}} \left| A_{wv} \right| \leq \sum_{v \in \mathsf{S}} \left| A_{wv}\right|.$
Define $A'$ in the following way:
\begin{eqnarray}
 A'_{st} & = & -\sum_{v \in T} \left|A_{wv}\right| \\
 A'_{tv}&=& 0, \ v \in \mathsf{B}\\
 A'_{tv}&=&A_{wv},  \ v \in \mathsf{T}\\
 A'_{sv}&=&A_{wv}, \ v \in \mathsf{S} \cup \mathsf{B}\\
 A'_{uv}&=&A_{uv}, \ uv \in \mathsf{E}'\setminus \left[\delta(s) \cup \delta(t)\right].
\end{eqnarray}
The inequality $A' \cdot P' \leq b$ is called the \emph{vertex splitting} 
of $A \cdot P \leq b$ with respect 
to $w \in \mathsf{V}$ and $\left(\mathsf{S}, \mathsf{T}, \mathsf{B}\right)$.
\end{dfn}

\begin{prop}
\label{prop:split}
Let graph $\mathsf{G}'$ and inequality $A' \cdot P' \leq b$ be 
vertex splittings of  $\mathsf{G}$ and $A \cdot P \leq b$ (resp.) with 
respect to 
$w \in \mathsf{V}$ and $\left(\mathsf{S}, \mathsf{T}, \mathsf{B}\right)$. 
If $A \cdot P \leq b$ is a valid inequality 
for $\mathrm{CUT}\left(\mathsf{G}\right)$, then 
 $A' \cdot P' \leq b$ is a valid inequality for 
$\mathrm{CUT}\left(\mathsf{G}'\right)$.
%
%  \st{If graph $\mathsf{G}'$ and inequality $A' \cdot P' \leq b$ are 
% vertex splittings of  $\mathsf{G}$ and $A \cdot P \leq b$ with respect to 
% $w \in \mathsf{V}$ and $\left(\mathsf{S}, \mathsf{T}, \mathsf{B}\right)$ 
% and $A \cdot P \leq b$ is a valid inequality 
% for $\mathrm{CUT}\left(\mathsf{G}\right)$, then 
%  $A' \cdot P' \leq b$ is a valid inequality for 
% $\mathrm{CUT}\left(\mathsf{G}'\right)$.}
\end{prop}

\begin{thm}
 The extended compatibility graph $\mathscr{G}$ and its suspension graph $\nabla \mathscr{G}$ can be obtained from the compatibility graph $\mathrm{G}$ and $\nabla \mathrm{G}$, respectively, using a sequence of vertex splitting operations.
 \label{thm:split}
\end{thm}

\begin{proof}
 Choose $x \in \mathrm{X}$ and let $C_1, \ldots , C_n$ be the contexts 
containing $x$. Then $\delta(x)$ contains \st{the} measurements in 
$\left[\cup_i C_i\right] \setminus C_1$. Starting with $\mathrm{G}$, the 
first operation is splitting $x$ into $x_1$ and $x_1'$ with respect 
to the partition \be\left(S_1 = C_1\setminus \left[\cup_{i>1} C_i\right] ,
 T_1 = \left[\cup_{i>1} C_i \right] \setminus C_1, B_1= \left[\cup_{i>1} 
C_i\right] \cap C_1\right).\ee Vertex  $x_1$ is connected to $S_1$, vertex 
$x_1'$ is connected to $T_1$ and both $x_1$ and $x_1'$ are connected to 
$B_1$. With this operation, we set $x_1$ as the copy of $x$ in 
$\mathscr{G}$ corresponding to context $C_1$. The next operation is split 
$x_1'$ into vertices $x_2$ and $x_2'$ with respect to partition \be 
\left(S_2=C_2\setminus\left[ \cup_{i>2}\right] C_i ,
 T_2=\left[\cup_{i>2} C_i\right] \setminus C_2,B_2=\left[\left[\cup_{i=2}^n 
C_i\right] \cap C_2 \right]\cup \left\{x_1\right\}\right).\ee With this 
operation, we set $x_2$ as the copy of $x$ in $\mathscr{G}$
 corresponding to context $C_2$. We proceed analogously, in each step 
splitting vertex $x_k'$ into $x_{k+1}$ and $x_{k+1}'$ with respect to the 
 partition 
 \begin{multline}\left( S_{k+1}=C_{k+1}\setminus \left[\cup_{i>{k+1}} C_i\right] , T_{k+1}=\left[\cup_{i>{k+1}} C_i \right]\setminus C_{k+1},\right. \\ \left.B_{k+1}=\left[\left[\cup_{i=k+1}^n C_i \right] \cap C_{k+1}\right]\cup \left\{x_1, \ldots , x_k\right\}\right).
 \end{multline} 
 With this chain of operations we eliminate vertex $x$
 and add the clique $x_1, \ldots , x_n$, each $x_i$ connected only to the vertices in context $C_i$ and the other $x_j$. Applying the same procedure
 to the other vertices in $X$ we recover $\mathscr{G}$.  A similar argument can be used for $\nabla \mathscr{G}$.
 \end{proof}

 A simple example of the procedure described in the previous proof is shown in Fig. \ref{fig:ex_splitting}.
 
 \begin{figure}[h!]
\begin{subfigure}{.33\columnwidth}
  \begin{center}
 \definecolor{qqccqq}{rgb}{0.5,0.95,0.5}
 \definecolor{zzccqq}{rgb}{0.7,0.9,0.2}
 \definecolor{wwccqq}{rgb}{0.4,0.8,0}
 \definecolor{qqwuqq}{rgb}{0,0.39215686274509803,0}
 \begin{tikzpicture}[scale=0.75,line cap=round,line join=round,>=triangle 45,x=1cm,y=1cm]
 \draw [line width=2pt] (1,4)-- (5,4);
 \draw [line width=2pt] (5,4)-- (3,2);
 \draw [line width=2pt] (3,2)-- (1,4);
 \draw [line width=2pt] (3,2)-- (1,0);
 \draw [line width=2pt] (1,0)-- (3,-2);
 \draw [line width=2pt] (3,-2)-- (5,0);
 \draw [line width=2pt] (5,0)-- (3,2);
 \draw [line width=2pt] (3,-2)-- (3,2);
 \draw [fill=qqwuqq] (3,2) circle (10pt);
 \draw[color=black] (3,2) node {$0$};
 \draw [fill=wwccqq] (1,4) circle (10pt);
 \draw[color=black] (1,4) node {$1$};
 \draw [fill=wwccqq] (5,4) circle (10pt);
 \draw[color=black] (5,4) node {$2$};
 \draw [fill=zzccqq] (1,0) circle (10pt);
 \draw[color=black] (1,0) node {$3$};
 \draw [fill=zzccqq] (5,0) circle (10pt);
 \draw[color=black] (5,0) node {$4$};
 \draw [fill=qqccqq] (3,-2) circle (10pt);
 \draw[color=black] (3,-2) node {$5$};
 \end{tikzpicture}
  \end{center}
  \caption{\footnotesize{Graph $\mathrm{G}$.}}
\end{subfigure}%
\begin{subfigure}{.33\columnwidth}
  \begin{center}
   \definecolor{qqccqq}{rgb}{0.5,0.95,0.5}
 \definecolor{zzccqq}{rgb}{0.7,0.9,0.2}
 \definecolor{wwccqq}{rgb}{0.4,0.8,0}
 \definecolor{qqwuqq}{rgb}{0,0.39215686274509803,0}
 \begin{tikzpicture}[scale=0.75,line cap=round,line join=round,>=triangle 45,x=1cm,y=1cm]
 \draw [line width=2pt] (1,4)-- (5,4);
 \draw [line width=2pt] (5,4)-- (3,2);
 \draw [line width=2pt] (3,2)-- (1,4);
 \draw [line width=2pt] (3,2)-- (1,0);
 \draw [line width=2pt] (1,0)-- (3,-2);
 \draw [line width=2pt] (3,-2)-- (5,0);
 %\draw [line width=2pt] (5,0)-- (3,2);
 \draw [line width=2pt] (0,-2)-- (3,-2);
  \draw [line width=2pt] (0,-2)-- (1,0);
    \draw [line width=2pt] (5,0)-- (1,0);
 \draw [fill=qqwuqq] (3,2) circle (10pt);
 \draw[color=black] (3,2) node {$0_1$};
 \draw [fill=wwccqq] (1,4) circle (10pt);
 \draw[color=black] (1,4) node {$1$};
 \draw [fill=wwccqq] (5,4) circle (10pt);
 \draw[color=black] (5,4) node {$2$};
 \draw [fill=qqwuqq] (1,0) circle (10pt);
 \draw[color=black] (1,0) node {$0_1'$};
 \draw [fill=zzccqq] (5,0) circle (10pt);
 \draw[color=black] (5,0) node {$4$};
 \draw [fill=qqccqq] (3,-2) circle (10pt);
 \draw[color=black] (3,-2) node {$5$};
 \draw [fill=zzccqq] (0,-2) circle (10pt);
 \draw[color=black] (0,-2) node {$3$};
 \end{tikzpicture}
  \end{center}
  \caption{\footnotesize{First splitting operation.}}
  \end{subfigure}
  \begin{subfigure}{.33\columnwidth}
  \begin{center}
   \definecolor{qqccqq}{rgb}{0.5,0.95,0.5}
 \definecolor{zzccqq}{rgb}{0.7,0.9,0.2}
 \definecolor{wwccqq}{rgb}{0.4,0.8,0}
 \definecolor{qqwuqq}{rgb}{0,0.39215686274509803,0}
 \begin{tikzpicture}[scale=0.75,line cap=round,line join=round,>=triangle 45,x=1cm,y=1cm]
 \draw [line width=2pt] (1,4)-- (5,4);
 \draw [line width=2pt] (5,4)-- (3,2);
\draw [line width=2pt] (3,2)-- (1,4);
 \draw [line width=2pt] (3,2)-- (1,0);
 \draw [line width=2pt] (1,0)-- (3,-2);
 \draw [line width=2pt] (3,-2)-- (5,0);
 \draw [line width=2pt] (5,0)-- (3,2);
 \draw [line width=2pt] (0,-2)-- (3,-2);
  \draw [line width=2pt] (0,-2)-- (1,0);
   \draw [line width=2pt] (6,-2)-- (3,-2);
  \draw [line width=2pt] (6,-2)-- (5,0);
    \draw [line width=2pt] (5,0)-- (1,0);
 \draw [fill=qqwuqq] (3,2) circle (10pt);
 \draw[color=black] (3,2) node {$0_1$};
 \draw [fill=wwccqq] (1,4) circle (10pt);
 \draw[color=black] (1,4) node {$1$};
 \draw [fill=wwccqq] (5,4) circle (10pt);
 \draw[color=black] (5,4) node {$2$};
 \draw [fill=qqwuqq] (1,0) circle (10pt);
 \draw[color=black] (1,0) node {$0_2$};
 \draw [fill=qqwuqq] (5,0) circle (10pt);
 \draw[color=black] (5,0) node {$0_3$};
 \draw [fill=qqccqq] (3,-2) circle (10pt);
 \draw[color=black] (3,-2) node {$5$};
 \draw [fill=zzccqq] (0,-2) circle (10pt);
 \draw[color=black] (0,-2) node {$3$};
  \draw [fill=zzccqq] (6,-2) circle (10pt);
 \draw[color=black] (6,-2) node {$4$};
 \end{tikzpicture}
  \end{center}
  \caption{\footnotesize{Second splitting operation.}}
  \end{subfigure}
  \caption{\footnotesize{(a) The compatibility graph of the scenario with 
measurements $0,\ldots , 5$ and contexts $C_1=\left\{0,1,2\right\}$, 
$C_2=\left\{0,3,5\right\}$
  and $C_3=\left\{0,4,5\right\}$. (b) Applying vertex splitting to vertex $0$ with respect to the partition $S_1= \left\{1,2\right\}$, $T_1=\left\{3,4,5\right\}$, $B_1=\emptyset$. Vertex $0_1$ is the copy of $0$ in $\mathscr{G}$ corresponding to context $C_1$. (c) Applying vertex splitting to vertex $0_1'$ with respect to the partition $S_2=
   \left\{3\right\}$, $T_1=\left\{4\right\}$, $B_1=\left\{0_1, 5\right\}$. This step generates vertices $0_2$ and $0_2'=0_3$, corresponding to contexts $C_2$ and $C_3$ 
   respectively. Applying a similar procedure to vertex $5$ we get 
$\mathscr{G}$.}  \label{fig:ex_splitting}}
\end{figure}

 Combining Prop. \ref{prop:split} and Thm. \ref{thm:split}, we have:
 
 \begin{cor}
 From valid inequalities for $\mathrm{CUT}\left(\nabla \mathrm{G}\right)$ we can generate necessary conditions for extended noncontextuality using 
 vertex splitting.
 \end{cor}
 
 \subsection{Triangular Elimination and Edge Contraction}
 
 %Combining the triangular elimination operation with a edge contraction operation we can also generate 
 %$\mathscr{G}$ from $\mathrm{G}$. 
 
 \begin{dfn}[Edge contraction for graphs]
 Let $\mathsf{G}=\left(\mathsf{V},\mathsf{E}\right)$ be a graph, $w 
\notin \mathsf{V}$, and $uv \in E$. The graph 
$\mathsf{G}'=\left(\mathsf{V}',\mathsf{E}'\right)$
 is a \emph{contraction} of $\mathsf{G}$ at edge $uv$ if $\mathsf{V}'= 
\left[\mathsf{V}\setminus \left\{u,v\right\}\right]\cup\left\{w\right\}$ 
and 
 $ \displaystyle \mathsf{E}'= \left[\mathsf{E} \setminus \left[ \left\{uv\right\}\cup  \left\{ux| x \in \delta(u)\right\} \cup  \left\{vx| x \in \delta(v)\right\}\right]
\right] \cup \left\{wx| x \in \delta(u) \cup \delta(v)\right\}.$
 \label{dfn:edge_contr_graphs}  
 \end{dfn}

 A simple example of this operation is shown in 
Fig.\ref{fig:contraction}.

  \begin{figure}[h!]
\begin{subfigure}{.5\columnwidth}
  \begin{center}
\definecolor{ccccff}{rgb}{0.8,0.8,1}
\definecolor{wwqqcc}{rgb}{0.4,0,0.8}
\begin{tikzpicture}[line cap=round,line join=round,>=triangle 45,x=1cm,y=1cm]
\draw [line width=2pt] (0,0)-- (0,2);
\draw [line width=2pt] (0,2)-- (2,1);
\draw [line width=2pt] (2,1)-- (0,0);
\draw [line width=2pt] (2,1)-- (4,1);
\draw [line width=2pt] (6,0)-- (4,1);
\draw [line width=2pt] (4,1)-- (6,2);
\draw [line width=2pt] (6,2)-- (6,0);
\draw [fill=wwqqcc] (0,2) circle (10pt);
\draw [fill=wwqqcc] (0,0) circle (10pt);
\draw [fill=ccccff] (2,1) circle (10pt);
\draw[color=black] (2,1) node {$0$};
\draw [fill=ccccff] (4,1) circle (10pt);
\draw[color=black] (4,1) node {$1$};
\draw [fill=wwqqcc] (6,2) circle (10pt);
\draw [fill=wwqqcc] (6,0) circle (10pt);
\end{tikzpicture}
  \end{center}
  \caption{\footnotesize{Graph $\mathsf{G}$.}}
\end{subfigure}%
\begin{subfigure}{.5\columnwidth}
  \begin{center}
\definecolor{ccccff}{rgb}{0.8,0.8,1}
\definecolor{wwqqcc}{rgb}{0.4,0,0.8}
\begin{tikzpicture}[line cap=round,line join=round,>=triangle 45,x=1cm,y=1cm]
\draw [line width=2pt] (10,0)-- (10,2);
\draw [line width=2pt] (10,2)-- (12,1);
\draw [line width=2pt] (12,1)-- (14,2);
\draw [line width=2pt] (14,2)-- (14,0);
\draw [line width=2pt] (14,0)-- (12,1);
\draw [line width=2pt] (12,1)-- (10,0);
\begin{scriptsize}
\draw [fill=wwqqcc] (10,2) circle (10pt);
\draw [fill=wwqqcc] (10,0) circle (10pt);
\draw [fill=ccccff] (12,1) circle (10pt);
%\draw[color=ccccff] (12.16,1.43) node {$I$};
\draw [fill=wwqqcc] (14,2) circle (10pt);
\draw [fill=wwqqcc] (14,0) circle (10pt);
\end{scriptsize}
\end{tikzpicture}
  \end{center}
  \caption{\footnotesize{Graph $\mathsf{G}'$.}}
  \end{subfigure}
  \caption{\footnotesize{Contraction of graph $\mathsf{G}$ at the edge 
connecting vertices $0$ and $1$.}  \label{fig:contraction}}
\end{figure}
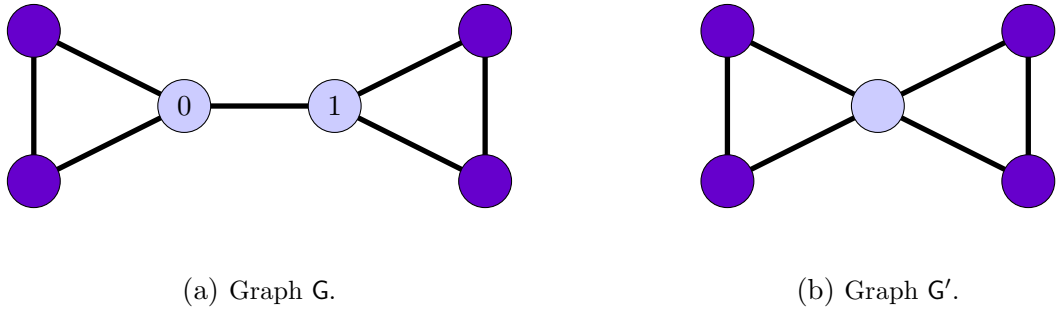

\begin{dfn}[Edge contraction for inequalities]
 Let $\mathsf{G}=\left(\mathsf{V}, \mathsf{E}\right)$ be a graph, $uv \in \mathsf{E}$
 and $A \cdot P \leq b$ be an inequality valid for $\mathrm{CUT}\left(\mathsf{G}\right)$.
Define $A'$ in the following way:
\begin{eqnarray}
 A'_{xy} & = & A_{xy}, \ x,y \neq w \\
 A'_{wx}&=& A_{ux}, \ x \in \delta(u) \setminus \delta(v)\\
 A'_{wx}&=&A_{vx},  \ x \in \delta(v) \setminus \delta(u)\\
 A'_{wx}&=&A_{ux}+ A_{vx}, \ x \in \delta(u) \cap 
\delta(v).
\end{eqnarray}
The inequality $A' \cdot P \leq b$ is called the \emph{contraction} of $A 
\cdot P \leq b$ at the edge $uv$.
\end{dfn}

\begin{prop}[Edge contraction lemma~\cite{AIT06,DL97}]
 If $\mathsf{G}'$ and $A' \cdot P \leq b$ are contractions of  $\mathsf{G}$ and $A \cdot P \leq b$, respectively,  at edge $uv$ and 
 $A \cdot P \leq b$  is a valid for $\mathrm{CUT}\left(\mathsf{G}\right)$, the inequality $A' \cdot P \leq b$ is valid for
 $\mathrm{CUT}\left(\mathsf{G}'\right)$.
\end{prop}

\begin{thm}
 The extended compatibility graph $\mathscr{G}$ ant its 
suspension graph $\nabla \mathscr{G}$ can be obtained from  $\mathrm{G}$ 
and $\nabla \mathrm{G}$, respectively, using triangular 
 elimination and edge contraction.
\end{thm}

\begin{proof}
 When some contexts have three elements or more, the problem with the construction of Thm. \ref{thm:extd_graph_elimintion} is that we have a copy for 
 $v \in X$ for each vertex in $\delta(v)$ instead of one copy for each context containing $v$. From this graph we can obtain $\mathrm{G}$ 
 identifying these extra copies contracting the corresponding edges. A similar argument can be used for $\nabla \mathscr{G}$.
\end{proof}

A simple example of this procedure is shown in Fig. \ref{fig:te+c}. As an  corollary, we have the following: 

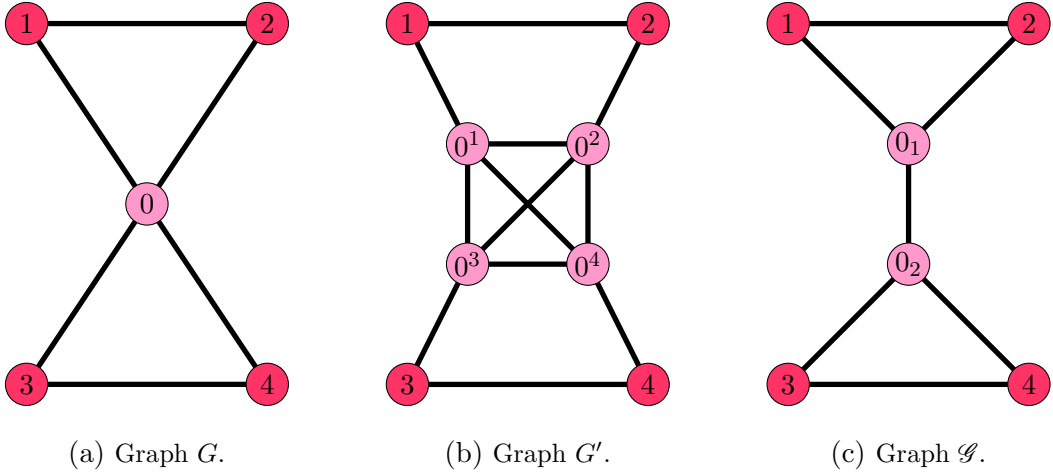
\begin{figure}
\begin{subfigure}{.3\columnwidth}
\definecolor{ffzzcc}{rgb}{1,0.6,0.8}
\definecolor{ffttww}{rgb}{1,0.2,0.4}
\begin{tikzpicture}[scale=0.8,line cap=round,line join=round,>=triangle 45,x=1cm,y=1cm, rotate=90]
\draw [line width=2pt] (-2,0)-- (-2,4);
\draw [line width=2pt] (-2,4)-- (1,2);
\draw [line width=2pt] (1,2)-- (-2,0);
\draw [line width=2pt] (1,2)-- (4,4);
\draw [line width=2pt] (4,4)-- (4,0);
\draw [line width=2pt] (4,0)-- (1,2);
\draw [fill=ffttww] (-2,4) circle (10pt);
\draw[color=black] (-2,4) node {$3$};
\draw [fill=ffttww] (-2,0) circle (10pt);
\draw[color=black] (-2,0) node {$4$};
\draw [fill=ffzzcc] (1,2) circle (10pt);
\draw[color=black] (1,2) node {$0$};
\draw [fill=ffttww] (4,4) circle (10pt);
\draw[color=black] (4,4) node {$1$};
\draw [fill=ffttww] (4,0) circle (10pt);
\draw[color=black] (4,0) node {$2$};
\end{tikzpicture}
\caption{\footnotesize{Graph $G$.}}
\end{subfigure}
\begin{subfigure}{.3\columnwidth}
\definecolor{ffzzcc}{rgb}{1,0.6,0.8}
\definecolor{ffttww}{rgb}{1,0.2,0.4}
\begin{tikzpicture}[scale=0.8,line cap=round,line join=round,>=triangle 45,x=1cm,y=1cm, rotate=90]
\draw [line width=2pt] (8,0)-- (10,1);
\draw [line width=2pt] (10,3)-- (8,4);
\draw [line width=2pt] (10,3)-- (12,3);
\draw [line width=2pt] (12,3)-- (12,1);
\draw [line width=2pt] (12,1)-- (10,1);
\draw [line width=2pt] (10,1)-- (10,3);
\draw [line width=2pt] (10,3)-- (12,1);
\draw [line width=2pt] (12,3)-- (10,1);
\draw [line width=2pt] (8,0)-- (8,4);
\draw [line width=2pt] (12,3)-- (14,4);
\draw [line width=2pt] (12,1)-- (14,0);
\draw [line width=2pt] (14,0)-- (14,4);
\draw [fill=ffttww] (8,4) circle (10pt);
\draw[color=black] (8,4) node {$3$};
\draw [fill=ffttww] (8,0) circle (10pt);
\draw[color=black] (8,0) node {$4$};
\draw [fill=ffzzcc] (10,3) circle (10pt);
\draw[color=black] (10,3) node {$0^3$};
\draw [fill=ffzzcc] (10,1) circle (10pt);
\draw[color=black] (10,1) node {$0^4$};
\draw [fill=ffzzcc] (12,3) circle (10pt);
\draw[color=black] (12,3) node {$0^1$};
\draw [fill=ffzzcc] (12,1) circle (10pt);
\draw[color=black] (12,1) node {$0^2$};
\draw [fill=ffttww] (14,4) circle (10pt);
\draw[color=black] (14,4) node {$1$};
\draw [fill=ffttww] (14,0) circle (10pt);
\draw[color=black] (14,0) node {$2$};
\end{tikzpicture}
\caption{\footnotesize{Graph $G'$.}}
\end{subfigure}
\begin{subfigure}{.3\columnwidth}
\definecolor{ffzzcc}{rgb}{1,0.6,0.8}
\definecolor{ffttww}{rgb}{1,0.2,0.4}
\begin{tikzpicture}[scale=0.8,line cap=round,line join=round,>=triangle 45,x=1cm,y=1cm, rotate=90]
\draw [line width=2pt] (5,-6)-- (3,-4);
\draw [line width=2pt] (3,-4)-- (3,-8);
\draw [line width=2pt] (3,-8)-- (5,-6);
\draw [line width=2pt] (5,-6)-- (7,-6);
\draw [line width=2pt] (7,-6)-- (9,-4);
\draw [line width=2pt] (9,-4)-- (9,-8);
\draw [line width=2pt] (9,-8)-- (7,-6);
\draw [fill=ffttww] (3,-4) circle (10pt);
\draw[color=black] (3,-4) node {$3$};
\draw [fill=ffttww] (3,-8) circle (10pt);
\draw[color=black] (3,-8) node {$4$};
\draw [fill=ffttww] (9,-8) circle (10pt);
\draw[color=black] (9,-8) node {$2$};
\draw [fill=ffttww] (9,-4) circle (10pt);
\draw[color=black] (9,-4) node {$1$};
\draw [fill=ffzzcc] (5,-6) circle (10pt);
\draw[color=black] (5,-6) node {$0_2$};
\draw [fill=ffzzcc] (7,-6) circle (10pt);
\draw[color=black] (7,-6) node {$0_1$};
\end{tikzpicture}
\caption{\footnotesize{Graph $\mathscr{G}$.}}
\end{subfigure}
\caption{\footnotesize{(a) Compatibility graph of the scenario with 
measurements $0, \ldots , 4$ and contexts $C_1=\left\{0,1,2\right\}$ and 
$C_2=\left\{0,3,4\right\}$.
(b) The graph $G'$ obtained from $G$ after applying the procedure described in Thm. \ref{thm:extd_graph_elimintion}. Notice that this is not the 
the extended compatibility graph of the scenario, since there are four copies of $0$ instead of two. (c) The extended compatibility graph of the scenario is obtained
contracting the edges $0^10^2$, which gives vertex $0_1$ (the copy of 
vertex $0$ corresponding to context $C_1$), and $0^30^4$, which gives 
vertex $0_2$ (the copy of vertex $0$ corresponding to context $C_2$).}}
\label{fig:te+c}
\end{figure}

\begin{cor}
 Valid inequalities for $\mathscr{G}$ can be generated combining triangular elimination and edge contraction
 of  valid inequalities for $\mathrm{G}$.
\end{cor}

This provides another tool to derive necessary conditions for extended noncontextuality in any scenario.

\section{The Peres-Mermin inequality}

Although the cut polytope provides a powerful tool to derive necessary conditions for contextuality, both in the standard and in the extended sense, 
it is not enough to characterize completely the set of noncontextual distributions in scenarios with contexts containing more then two random variables, since  there are
contextual behaviors that can not be detected when we look only to the
binary expectation values of Eq. \eqref{eq:p2}, that is, there are 
contextual  behaviors $B$ for which 
$P_B\in \mathrm{CUT}\left( \nabla G\right)$ \cite{GWAN11}. 

With this in mind, it would be useful to find strategies to derive necessary conditions for extended contextuality from inequalities that involve expectation values 
with more than two random variables. In what follows, we show that this is 
possible with  a simple procedure, similar to triangular elimination, using 
the Peres-Mermin inequality as an example. 

The Peres-Mermin square is a contextuality scenario with nine measurements 
$A_i$, $i=1, \ldots 9$, with outcomes $\pm 1$, and compatibility hypergraph 
shown in Fig. \eqref{fig:Peres_Mermin}. These measurements can be chosen 
in quantum theory in such a way that the product of the three measurements 
in each line and in the first two columns is equal to the 
identity operator $I$, while the product of the measurements in the last 
column is equal to $-I$. 

\begin{figure}[h!]
  \begin{center}
\definecolor{yqqqqq}{rgb}{0.5019607843137255,0,0}
\definecolor{qqzzff}{rgb}{0,0.6,1}
\definecolor{wwqqcc}{rgb}{0.6,0,1}
\definecolor{qqffqq}{rgb}{0,1,0}
\definecolor{ffttww}{rgb}{1,0.2,0.4}
\definecolor{qqffff}{rgb}{0,1,1}
\definecolor{ffwwqq}{rgb}{1,0.4,0}
\definecolor{ffdxqq}{rgb}{1,0.8431372549019608,0}
\definecolor{qqwuqq}{rgb}{0,0.39215686274509803,0}
\definecolor{ffqqqq}{rgb}{1,0,0}
\definecolor{qqqqff}{rgb}{0.2,0.3,1}
\begin{tikzpicture}[scale=0.8,line cap=round,line join=round,>=triangle 45,x=1cm,y=1cm]
\draw [line width=1pt,dash pattern=on 1pt off 3pt,color=qqzzff] (-1,2)-- (-1,0);
\draw [line width=1pt,dash pattern=on 1pt off 3pt,color=qqzzff] (-1,0)-- (9.020188903166318,0.009276223793417249);
\draw [line width=1pt,dash pattern=on 1pt off 3pt,color=qqzzff] (9.020188903166318,0.009276223793417249)-- (9,2);
\draw [line width=1pt,dash pattern=on 1pt off 3pt,color=qqzzff] (9,2)-- (-1,2);
\draw [line width=1pt,dash pattern=on 1pt off 3pt,color=qqzzff] (-1,-1)-- (-1,-3);
\draw [line width=1pt,dash pattern=on 1pt off 3pt,color=qqzzff] (-1,-3)-- (9,-3);
\draw [line width=1pt,dash pattern=on 1pt off 3pt,color=qqzzff] (9,-3)-- (9,-1);
\draw [line width=1pt,dash pattern=on 1pt off 3pt,color=qqzzff] (9,-1)-- (-1,-1);
\draw [line width=1pt,dash pattern=on 1pt off 3pt,color=qqzzff] (-1,-4)-- (-1,-6);
\draw [line width=1pt,dash pattern=on 1pt off 3pt,color=qqzzff] (-1,-6)-- (9,-6);
\draw [line width=1pt,dash pattern=on 1pt off 3pt,color=qqzzff] (9,-6)-- (9,-4);
\draw [line width=1pt,dash pattern=on 1pt off 3pt,color=qqzzff] (9,-4)-- (-1,-4);
\draw [line width=1pt,dash pattern=on 1pt off 3pt,color=qqzzff] (0,3)-- (2,3);
\draw [line width=1pt,dash pattern=on 1pt off 3pt,color=qqzzff] (2,3)-- (2,-7);
\draw [line width=1pt,dash pattern=on 1pt off 3pt,color=qqzzff] (2,-7)-- (0,-7);
\draw [line width=1pt,dash pattern=on 1pt off 3pt,color=qqzzff] (0,-7)-- (0,3);
\draw [line width=1pt,dash pattern=on 1pt off 3pt,color=qqzzff] (3,3)-- (5,3);
\draw [line width=1pt,dash pattern=on 1pt off 3pt,color=qqzzff] (5,3)-- (5,-7);
\draw [line width=1pt,dash pattern=on 1pt off 3pt,color=qqzzff] (5,-7)-- (3,-7);
\draw [line width=1pt,dash pattern=on 1pt off 3pt,color=qqzzff] (3,-7)-- (3,3);
\draw [line width=1pt,dash pattern=on 1pt off 3pt,color=yqqqqq] (6,3)-- (8,3);
\draw [line width=1pt,dash pattern=on 1pt off 3pt,color=yqqqqq] (8,3)-- (8,-7);
\draw [line width=1pt,dash pattern=on 1pt off 3pt,color=yqqqqq] (8,-7)-- (6,-7);
\draw [line width=1pt,dash pattern=on 1pt off 3pt,color=yqqqqq] (6,-7)-- (6,3);
\begin{large}
\draw [qqqqff,fill=qqqqff] (1,1) circle (20pt);
\draw[color=black] (1,1) node {$A_1$};
\draw [ffqqqq, fill=ffqqqq] (4,1) circle (20pt);
\draw[color=black] (4,1) node {$A_2$};
\draw [fill=qqwuqq, qqwuqq] (7,1) circle (20pt);
\draw[color=black] (7,1) node {$A_3$};
\draw [ffdxqq,fill=ffdxqq] (1,-2) circle (20pt);
\draw[color=black] (1,-2) node {$A_4$};
\draw [ffwwqq,fill=ffwwqq] (4,-2) circle (20pt);
\draw[color=black] (4,-2) node {$A_5$};
\draw [fill=qqffff, qqffff] (7,-2) circle (20pt);
\draw[color=black] (7,-2) node {$A_6$};
\draw [fill=ffttww, ffttww] (1,-5) circle (20pt);
\draw[color=black] (1,-5) node {$A_7$};
\draw [fill=qqffqq, qqffqq] (4,-5) circle (20pt);
\draw[color=black] (4,-5) node {$A_8$};
\draw [fill=wwqqcc, wwqqcc] (7,-5) circle (20pt);
\draw[color=black] (7,-5) node {$A_9$};
\end{large}
\end{tikzpicture}
  \end{center}
  \caption{\footnotesize{Compatibility hypergraph $\mathrm{H}$ of the 
Peres-Mermin  scenario.}}
  \label{fig:Peres_Mermin}
\end{figure}

For this scenario, every noncontextual behavior must satisfy the inequality
\be\left\langle A_1A_2A_3\right\rangle + \left\langle A_4A_5A_6\right\rangle + \left\langle A_7A_8A_9\right\rangle +
\left\langle A_1A_4A_7\right\rangle + \left\langle A_2A_5A_8\right\rangle - \left\langle A_3A_6A_9\right\rangle  \leq  4\label{eq:peres-mermim}\ee
while for all quantum behaviors the left hand side is equal to $6$. This is one of the famous examples of 
\emph{state independent contextuality}: for this choice of measurements, all quantum states yield noncontextual behaviors.

\begin{figure}[h!]
  \begin{center}
  \definecolor{yqqqqq}{rgb}{0.5019607843137255,0,0}
  \definecolor{qqzzff}{rgb}{0,0.6,1}
  \definecolor{qqffqq}{rgb}{0,1,0}
  \definecolor{wwqqcc}{rgb}{0.6,0,1}
  \definecolor{ttffqq}{rgb}{0.2,1,0}
  \definecolor{ffttww}{rgb}{1,0.2,0.4}
  \definecolor{qqffff}{rgb}{0,1,1}
  \definecolor{ffxfqq}{rgb}{1,0.4980392156862745,0}
  \definecolor{ffdxqq}{rgb}{1,0.8431372549019608,0}
  \definecolor{qqwuqq}{rgb}{0,0.39215686274509803,0}
  \definecolor{ffqqqq}{rgb}{1,0,0}
  \definecolor{qqqqff}{rgb}{0.2,0.3,1}
  \begin{tikzpicture}[scale=0.75,line cap=round,line join=round,>=triangle 45,x=1cm,y=1cm]
\draw [line width=1pt,dash pattern=on 1pt off 3pt,color=qqzzff] (1,4)-- (1,2);
\draw [line width=1pt,dash pattern=on 1pt off 3pt,color=qqzzff] (1,2)-- (7,2);
\draw [line width=1pt,dash pattern=on 1pt off 3pt,color=qqzzff] (7,2)-- (7,4);
\draw [line width=1pt,dash pattern=on 1pt off 3pt,color=qqzzff] (7,4)-- (1,4);
\draw [line width=1pt,dash pattern=on 1pt off 3pt,color=qqzzff] (8,4)-- (8,2);
\draw [line width=1pt,dash pattern=on 1pt off 3pt,color=qqzzff] (14,2)-- (8,2);
\draw [line width=1pt,dash pattern=on 1pt off 3pt,color=qqzzff] (1,1)-- (1,-1);
\draw [line width=1pt,dash pattern=on 1pt off 3pt,color=qqzzff] (1,-1)-- (7,-1);
\draw [line width=1pt,dash pattern=on 1pt off 3pt,color=qqzzff] (7,-1)-- (7,1);
\draw [line width=1pt,dash pattern=on 1pt off 3pt,color=qqzzff] (7,1)-- (1,1);
\draw [line width=1pt,dash pattern=on 1pt off 3pt,color=qqzzff] (8,1)-- (8,-1);
\draw [line width=1pt,dash pattern=on 1pt off 3pt,color=qqzzff] (8,-1)-- (14,-1);
\draw [line width=1pt,dash pattern=on 1pt off 3pt,color=qqzzff] (14,1)-- (8,1);
\draw [line width=1pt,dash pattern=on 1pt off 3pt,color=qqzzff] (14,-1)-- (14,1);
\draw [line width=1pt,dash pattern=on 1pt off 3pt,color=qqzzff] (14,2)-- (14,4);
\draw [line width=1pt,dash pattern=on 1pt off 3pt,color=qqzzff] (14,4)-- (8,4);
\draw [line width=1pt,dash pattern=on 1pt off 3pt,color=qqzzff] (15,4)-- (15,2);
\draw [line width=1pt,dash pattern=on 1pt off 3pt,color=yqqqqq] (15,1)-- (15,-1);
\draw [line width=1pt,dash pattern=on 1pt off 3pt,color=yqqqqq] (15,-1)-- (21,-1);
\draw [line width=1pt,dash pattern=on 1pt off 3pt,color=yqqqqq] (21,-1)-- (21,1);
\draw [line width=1pt,dash pattern=on 1pt off 3pt,color=yqqqqq] (21,1)-- (15,1);
\draw [line width=1pt,dash pattern=on 1pt off 3pt,color=qqzzff] (15,2)-- (21,2);
\draw [line width=1pt,dash pattern=on 1pt off 3pt,color=qqzzff] (21,2)-- (21,4);
\draw [line width=1pt,dash pattern=on 1pt off 3pt,color=qqzzff] (21,4)-- (15,4);
\draw [line width=1pt,dash pattern=on 1pt off 3pt] (2,3)-- (2,0);
\draw [line width=1pt,dash pattern=on 1pt off 3pt] (4,3)-- (9,0);
\draw [line width=1pt,dash pattern=on 1pt off 3pt] (6,3)-- (16,0);
\draw [line width=1pt,dash pattern=on 1pt off 3pt] (16,3)-- (6,0);
\draw [line width=1pt,dash pattern=on 1pt off 3pt] (18,3)-- (13,0);
\draw [line width=1pt,dash pattern=on 1pt off 3pt] (11,3)-- (11,0);
\draw [line width=1pt,dash pattern=on 1pt off 3pt] (13,3)-- (18,0);
\draw [line width=1pt,dash pattern=on 1pt off 3pt] (20,3)-- (20,0);
\draw [line width=1pt,dash pattern=on 1pt off 3pt] (9,3)-- (4,0);
\draw [fill=qqqqff] (2,3) circle (18pt);
\draw[color=black] (2,3) node {$A_1^1$};
\draw [fill=ffqqqq] (4,3) circle (18pt);
\draw[color=black] (4,3) node {$A_2^1$};
\draw [fill=qqwuqq] (6,3) circle (18pt);
\draw[color=black] (6,3) node {$A_3^1$};
\draw [fill=ffdxqq] (9,3) circle (18pt);
\draw[color=black] (9,3) node {$A_4^2$};
\draw [fill=ffxfqq] (11,3) circle (18pt);
\draw[color=black] (11,3) node {$A_5^2$};
\draw [fill=qqffff] (13,3) circle (18pt);
\draw[color=black] (13,3) node {$A_6^2$};
\draw [fill=ffttww] (16,3) circle (18pt);
\draw[color=black] (16,3) node {$A_7^3$};
\draw [fill=ttffqq] (18,3) circle (18pt);
\draw[color=black] (18,3) node {$A_8^3$};
\draw [fill=wwqqcc] (20,3) circle (18pt);
\draw[color=black] (20,3) node {$A_9^3$};
\draw [fill=qqqqff] (2,0) circle (18pt);
\draw[color=black] (2,0) node {$A_1^4$};
\draw [fill=ffdxqq] (4,0) circle (18pt);
\draw[color=black] (4,0) node {$A_4^4$};
\draw [fill=ffttww] (6,0) circle (18pt);
\draw[color=black] (6,0) node {$A_7^4$};
\draw [fill=ffqqqq] (9,0) circle (18pt);
\draw[color=black] (9,0) node {$A_2^5$};
\draw [fill=ffxfqq] (11,0) circle (18pt);
\draw[color=black] (11,0) node {$A_5^5$};
\draw [fill=qqffqq] (13,0) circle (18pt);
\draw[color=black] (13,0) node {$A_8^5$};
\draw [fill=qqwuqq] (16,0) circle (18pt);
\draw[color=black] (16,0) node {$A_3^6$};
\draw [fill=qqffff] (18,0) circle (18pt);
\draw[color=black] (18,0) node {$A_6^6$};
\draw [fill=wwqqcc] (20,0) circle (18pt);
\draw[color=black] (20,0) node {$A_9^6$};
\end{tikzpicture}
  \end{center}
  \caption{\footnotesize{Extended compatibility hypergraph $\mathscr{H}$ 
of the Peres-Mermin contextuality scenario.}}
  \label{fig:Peres_Mermin_Ext}
\end{figure}

The extended compatibility hypergraph for this scenario is shown in Fig. 
\ref{fig:Peres_Mermin_Ext}. Labeling the hyperedges 
of $\mathrm{H}$ defined by the rows in Fig. \ref{fig:Peres_Mermin} as $1, 
2,3$ and the hyperedges defined by the columns as $4,5,6$,
each measurement $A_i$ is divided in two new vertices of $\mathscr{H}$ $A_i^j$ and $A_i^k$, where $j\in \{1,2,3\}$ and $k\in \{4,5,6\}$ according
to the row and column $A_i$ belongs to. Although the tools provided by the 
$\mathrm{CUT}$ polytope can not be used in this case, since the inequality \eqref{eq:peres-mermim} involves mean values of the product of
three measurements instead of two, some ideas of Sec. \ref{sec:valid_ineq} can be used in similar way  to derive valid inequalities for  
the extended scenario from it.

We start with the Ineq. \ref{eq:peres-mermim}, substituting each $A_i$ with its copy $A_i^j$ with $j\in \{1,2,3\}$:

\be \left\langle A_1^1A_2^1A_3^1\right\rangle + \left\langle A_4^2A_5^2A_6^2\right\rangle + \left\langle A_7^3A_8^3A_9^3\right\rangle +
\left\langle A_1^1A_4^2A_7^3\right\rangle + \left\langle A_2^1A_5^2A_8^3\right\rangle - \left\langle A_3^1A_6^2A_9^3\right\rangle  \leq  4\ee
valid for all noncontextual extended behaviors.

To eliminate the term $\left\langle A_1^1A_4^2A_7^3\right\rangle$
 we use 
 \be A_1^1A_4^2A_7^3 = A_1^4A_4^4A_7^4 +\Delta A_1 A_4^4 A_7^4 + A_1^1\Delta A_4 A_7^4 +A_1^1A_4^2\Delta A_7\ee
 where $\Delta A_1 = A_1^1 - A_1^4$ and similar for  $\Delta A_4$ and  $\Delta A_7$.
 From this we get
 \begin{eqnarray} \left\langle A_1^1A_2^1A_3^1\right\rangle + \left\langle A_4^2A_5^2A_6^2\right\rangle + \left\langle A_7^3A_8^3A_9^3\right\rangle +
\left\langle A_1^4A_4^4A_7^4\right\rangle + \left\langle A_2^1A_5^2A_8^3\right\rangle - \left\langle A_3^1A_6^2A_9^3\right\rangle & \leq & \nonumber \\
4 - \left\langle \Delta A_1 A_4^4 A_7^4  \right\rangle -  \left\langle A_1^1\Delta A_4 A_7^4 \right\rangle - \left\langle A_1^1A_4^2\Delta A_7\right\rangle&\leq &\\
 4+ \sum_{i=1}^4\left|\Delta A_i\right|& &\end{eqnarray}
 
Proceeding analogously with the other terms, we get the inequality
 \be\left\langle A_1^1A_2^1A_3^1\right\rangle + \left\langle A_4^2A_5^2A_6^2\right\rangle + \left\langle A_7^3A_8^3A_9^3\right\rangle +
\left\langle A_1^4A_4^4A_7^4\right\rangle + \left\langle A_2^5A_5^5A_8^5\right\rangle - \left\langle A_3^6A_6^6A_9^6\right\rangle  \leq 
 4+ \sum_{i=1}^9\left|\Delta A_i\right|\ee
 valid for all noncontextual extended behaviors. This inequality is tight 
and reduces to the original Peres-Mermin Ineq. \eqref{eq:peres-mermim} for 
non-disturbing
 behaviors.
 
% \subsection{Experimental Demonstration of Extended Contextuality for the 
%Peres-Mermin Scenario}

\section{Discussion}
\label{sec:discussion}
Apart from its primal importance in the foundations of quantum physics, 
contextuality has been discovered as a potential resource for quantum 
computing  \cite{Raussendorf13, HWVE14, DGBR14}, random number certification 
\cite{UZZWYDDK13}, and several other tasks in the particular case of Bell 
scenarios \cite{BCPSW13}.
% found several applications as a resource in quantum cryptography \cite{Acin2007}, randomness generation/amplification \cite{Pironio2010,Colbeck2012}, self-testing \cite{Mayers2004} and distributed computing \cite{Buhrman2010}.
 Within these both fundamental and applied perspectives, certifying contextuality  experimentally is undoubtedly an important primitive. It is then crucial  to develop a robust theoretical framework for contextuality that can be easily applied to real experiments. This should include the possibility of treating
  sets of random variables that do not satisfy the assumption of  
\emph{non-disturbance}, which will be hardly satisfied in  experimental 
implementations~\cite{KDL15}.

Here we have further developed the extended definition of noncontextuality of Ref. \cite{KDL15}, which can be applied in situations where the 
non-distrubance condition does not hold, rewriting it in graph-theoretical 
terms. We then explore the geometrical aspects of the graph approach to 
contextuality to derive necessary conditions for extended contextuality 
that can be tested directly with  experimental data in any 
 contextuality experiment and which reduce to traditional necessary conditions for noncontextuality if the non-disturbance condition is satisfied. 

It would be interesting to give a characterization of which of these 
inequalities are facet-defining. In Ref. \cite{AIT06}, several results 
regarding this issue were proved, but unfortunately our scenarios do not 
satisfy the hypotheses needed for the validity of such results.
 A more ambitious problem would be  to identify which scenarios can be
completely characterized with these procedures, the $n$-cycle scenarios 
being an important example. We leave these inquiries for future work,
hoping that our results might motivate further research in these directions.

\begin{acknowledgments}
The Authors thank  Jan-\AA{}ke Larsson and Ad\'an Cabello for valuable discussions.
This work was done during the Post-doctoral Summer Program of Instituto de Matem\'atica Pura e Aplicada (IMPA) 2017. 
BA and CD thank IMPA for its support and 
hospitality.
BA acknowledges financial support from the Brazilian ministries MEC and 
MCTIC and CNPq. CD acknowledges financial support from CAPES and CNPq. 
RO aknowlodges  the financial support of
Bolsa de Produtividade  em Pesquisa
from CNPq.
\end{acknowledgments}

\bibliography{biblio}

\end{document}